\documentclass[12pt]{article}

\usepackage[margin=1in]{geometry}  	
\usepackage{amsfonts}
\usepackage{amsmath} 
\usepackage{booktabs}
\usepackage{multirow}
\usepackage{bbm}

\usepackage{relsize}

\usepackage[doublespacing]{setspace} 
\usepackage{float}
\usepackage{caption}
\usepackage{graphicx}
\usepackage{adjustbox}
\usepackage{lscape}

\usepackage{rotating}

\usepackage{adjustbox}

\usepackage{xcolor}

\usepackage{bbm}
\usepackage{subfig}

\makeatletter
\newcommand*{\indep}{%
  \mathbin{%
    \mathpalette{\@indep}{}%
  }%
}
\newcommand*{\nindep}{%
  \mathbin{%                   % The final symbol is a binary math operator
    \mathpalette{\@indep}{\not}% \mathpalette helps for the adaptation
  }%
}
\newcommand*{\@indep}[2]{%
  \sbox0{$#1\perp\m@th$}%        box 0 contains \perp symbol
  \sbox2{$#1=$}%                 box 2 for the height of =
  \sbox4{$#1\vcenter{}$}%        box 4 for the height of the math axis
  \rlap{\copy0}%                 first \perp
  \dimen@=\dimexpr\ht2-\ht4-.2pt\relax
  \kern\dimen@
  {#2}%
  \kern\dimen@
  \copy0 %                       second \perp
} 
\makeatother

\usepackage{tikz}
\usetikzlibrary{positioning,shapes.geometric}
\usetikzlibrary{arrows}

\usepackage{cite}

\DeclareMathOperator{\E}{\textnormal{\mbox{E}}}

\usepackage[utf8]{inputenc}
\DeclareUnicodeCharacter{200E}{}

\DeclareMathOperator*{\argmax}{arg\,max}

\definecolor{forestgreen}{RGB}{34,139,34}
\usepackage{hyperref}
\hypersetup{
    colorlinks=true,
    linkcolor=magenta,
    filecolor=magenta, 
    citecolor=forestgreen,      
    urlcolor=blue
}

\usepackage[utf8]{inputenc}
\DeclareUnicodeCharacter{200E}{}

\usepackage{mathrsfs}

\usepackage{datetime}

\usepackage{amsthm}
\newtheorem{theorem}{Theorem}
\newtheorem{proposition}[theorem]{Proposition}
\newtheorem*{remark}{Remark}

\usepackage{xpatch}
\makeatletter
\xpatchcmd{\proof}{\@addpunct{.}}{\@addpunct{:}}{}{}
\makeatother

\usepackage{authblk}

\makeatletter
\def\@seccntformat#1{\@ifundefined{#1@cntformat}%
   {\csname the#1\endcsname\quad}  % default
   {\csname #1@cntformat\endcsname}% enable individual control
}
\let\oldappendix\appendix %% save current definition of \appendix
\renewcommand\appendix{%
    \oldappendix
    \newcommand{\section@cntformat}{\appendixname~\thesection\quad}
}
\makeatother

\usepackage[absolute,showboxes]{textpos}

\TPMargin{5pt}

\newcommand{\copyrightstatement}{
    \begin{textblock}{0.84}(0.08,0.93)    % tweak here: {box width}(leftposition, rightposition)
         \noindent
         \footnotesize
         This draft manuscript presents work in progress. \\
         Comments and reports of mistakes are very much welcome at \href{mailto:issa\_dahabreh@brown.edu}{issa\_dahabreh@brown.edu}.
    \end{textblock}
}

\def\paperversionmajor{12}
\def\paperversionminor{0}

\begin{document}

\title{Generalizing trial findings using nested trial designs with sub-sampling of non-randomized individuals}

\author[1,2,3]{Issa J. Dahabreh
%\thanks{Address for correspondence: Dr. Issa J. Dahabreh; Box G-S121-8; Brown University, Providence, RI 02912; email: \texttt{issa\_dahabreh@brown.edu}; phone: 401-863-XXXX.}
}
\author[3,4,5]{Miguel A. Hern\'an}
\author[1]{Sarah E. Robertson}
\author[6]{Ashley Buchanan}
\author[7]{Jon A. Steingrimsson}

\affil[1]{Center for Evidence Synthesis in Health and Department of Health Services, Policy \& Practice, Brown University, Providence, RI}
\affil[2]{Department of Epidemiology, Brown University, Providence, RI}
\affil[3]{Department of Epidemiology, Harvard T.H. Chan School of Public Health, Boston, MA}
\affil[4]{Department of Biostatistics, Harvard School of Public Health, Boston, MA}
\affil[5]{Harvard-MIT Division of Health Sciences and Technology, Boston, MA}
\affil[6]{Department of Pharmacy Practice, College of Pharmacy, University of Rhode Island, RI}
\affil[7]{Department of Biostatistics, Brown University School of Public Health, Providence, RI}

\copyrightstatement

\maketitle{}

\thispagestyle{empty}

%%%%%%%%%%%%%%%%%%%%%%%%%%%%%%%%%%%%%%%%%%%%%%%%%%%%%%%%%%%%%%%
\clearpage
\thispagestyle{empty}
\vspace*{1.6in}
\begin{abstract}
\noindent
To generalize inferences from a randomized trial to the target population of all trial-eligible individuals, investigators can use nested trial designs, where the randomized individuals are nested within a cohort of trial-eligible individuals, including those who are not offered or refuse randomization. In these designs, data on baseline covariates are collected from the entire cohort, and treatment and outcome data need only be collected from randomized individuals. In this paper, we describe nested trial designs that improve research economy by collecting additional baseline covariate data after sub-sampling non-randomized individuals (i.e., a two-stage design), using sampling probabilities that may depend on the initial set of baseline covariates available from all individuals in the cohort. We propose an estimator for the potential outcome mean in the target population of all trial-eligible individuals and show that our estimator is doubly robust, in the sense that it is consistent when either the model for the conditional outcome mean among randomized individuals or the model for the probability of trial participation is correctly specified. We assess the impact of sub-sampling on the asymptotic variance of our estimator and examine the estimator's finite-sample performance in a simulation study. We illustrate the methods using data from the Coronary Artery Surgery Study (CASS).
\end{abstract}
%%%%%%%%%%%%%%%%%%%%%%%%%%%%%%%%%%%%%%%%%%%%%%%%%%%%%%%%%%%%%%%

%%%%%%%%%%%%%%%%%%%%%%%%%%%%%%%%%%%%%%%%%%%%%%%%%%%%%%%%%%%%%%%
\clearpage
\section{Background}
\setcounter{page}{1}
%%%%%%%%%%%%%%%%%%%%%%%%%%%%%%%%%%%%%%%%%%%%%%%%%%%%%%%%%%%%%%%

Among individuals invited to participate in a randomized trial, those who agree to be randomized often differ from those who decline in terms of variables that are modifiers of the treatment effect. When that is the case, potential (counterfactual) outcome means and average treatment effects estimated in the trial do not directly apply to the \emph{target population of all trial-eligible individuals}. To address this problem, investigators can use a \emph{nested trial design} \cite{dahabreh2018generalizing}, where the randomized individuals are nested within a cohort of trial-eligible individuals, including those who are not offered or refuse randomization. In this design, baseline covariate data are collected from all individuals in the cohort, but treatment and outcome data need only be collected from randomized individuals. This is the basic study design approach in comprehensive cohort studies \cite{olschewski1985, olschewski1992} and pragmatic randomized trials conducted within large health-care systems \cite{choudhry2017}. 

When randomized and non-randomized individuals in nested trial designs are exchangeable conditional on baseline covariates \cite{cole2010,OMuircheartaigh2014, tipton2012, hartman2013, lesko2017practical}, we recently proposed efficient and robust estimators for the potential (counterfactual) outcome means and the average treatment effect \cite{dahabreh2018generalizing} in the target population. The validity of the estimators depends on the conditional exchangeability of randomized and non-randomized individuals, thus, it is important to have information on a rich enough set of baseline covariates, both from randomized and non-randomized individuals, to render the condition plausible.

In many applications, the baseline covariates that can be easily collected from non-randomized individuals are only a subset of the covariates collected from randomized individuals. The common baseline covariates collected from both randomized and non-randomized individuals may be insufficient for exchangeability and valid inference requires the collection of additional covariate information from the non-randomized individuals. When data collection is expensive, research economy can be improved by using a two-stage design \cite{rose2011targeted} with \emph{sub-sampling of non-randomized individuals}, that is, by collecting additional baseline covariate information only among a subset of the non-randomized individuals. The subset of non-randomized individuals targeted for additional data collection may be selected using sampling probabilities that depend on the initial set of baseline covariates.

In this paper we examine methods for generalizing causal inferences in nested trial designs with sub-sampling of non-randomized individuals (i.e., two-stage designs), when the sampling probabilities depend on the initial set of baseline auxiliary covariates. We propose an efficient and robust estimator for the potential outcome means in the target population of all trial-eligible individuals and show that our estimator is doubly robust, in the sense that it is consistent when either the model for the conditional outcome mean among randomized individuals or the model for the probability of trial participation is correctly specified. We assess the impact of sub-sampling on the asymptotic variance of the estimator and examine its finite-sample performance in a simulation study. We illustrate the application of the methods using data from the Coronary Artery Surgery Study (CASS) \cite{william1983}.

%%%%%%%%%%%%%%%%%%%%%%%%%%%%%%%%%%%%%%%%%%%%%%%%%%%%%%%%%%%%%%%
\section{Study designs and identifiability conditions}
%%%%%%%%%%%%%%%%%%%%%%%%%%%%%%%%%%%%%%%%%%%%%%%%%%%%%%%%%%%%%%%

\subsection{Nested trial designs with sub-sampling}

We begin by considering \emph{nested trial designs}, where a randomized trial is nested in a cohort of trial-eligible individuals \cite{dahabreh2018generalizing}. Let $A$ be the assigned treatment that takes values in $\mathcal A$, the set of treatments assessed in the randomized trial (we only consider discrete treatments); $Y$ the observable outcome; $X$ the (possibly high dimensional) vector of baseline covariates; and $S$ an indicator for trial participation ($S=1$ for randomized individuals; $S=0$ for non-randomized individuals).

When baseline covariate data are collected from all individuals in the cohort, but treatment and outcome data are collected only from randomized individuals, the observed data are 
\begin{equation*}
    \begin{cases}
      (X, S = 1, A, Y) \mbox{, for randomized individuals;} \\
      (X, S =0) \mbox{, for non-randomized individuals.}
    \end{cases}
\end{equation*}

Now, suppose that the baseline covariates are partitioned as $X = (X_1, X_2)$ and that the component $X_1$ is readily available, whereas $X_2$ is expensive to collect. For example, suppose that a randomized trial is nested within a health-care system and that trial-eligible individuals in the health-care system can be identified using routinely collected data. Claims and electronic health record data ($X_1$) are available from both randomized and non-randomized individuals at very low cost. In contrast, specialized laboratory/ imaging test results, or interview data ($X_2$), which are collected from randomized individuals, may be unavailable in the routinely collected data, expensive to collect (e.g., requiring manual chart abstraction), and necessary to ensure randomized and non-randomized individuals are exchangeable (see below). 

In settings like this, to avoid collecting the expensive covariates on all non-randomized individuals, it is natural to consider \emph{a two-stage design} where we sample non-randomized individuals for additional data collection, with sampling probabilities that may depend on $X_1$. We refer to this design as a \emph{nested trial with sub-sampling of non-randomized individuals.} For example, \cite{buchanan2018generalizing} described a special case of this design where the sampling probability was not allowed to depend on baseline covariates. Following \cite{buchanan2018generalizing} and related work on two-stage designs (e.g., \cite{weinberg1990design, weinberg1991randomized, breslow2000semi, kulich2004improving}), we assume assume Bernoulli-type (independent) sampling \cite{breslow2007weighted, saegusa2013weighted} of non-randomized individuals.

Let $D = 1$ be an indicator for whether $X_2$ data are collected from an individual; $D = 1$ for randomized individuals and sampled non-randomized individuals; $D=0$ for non-sampled non-randomized individuals. Using this notation, the observed data from a nested trial with sub-sampling of non-randomized individuals are
\begin{equation*}
    \begin{cases}
      (S = 1, D = 1, X, A, Y) & \mbox{, for randomized individuals;} \\
      (S = 0, D = 1, X) & \mbox{, for sampled non-randomized individuals;} \\
      (S = 0, D = 0, X_1) & \mbox{, for non-sampled non-randomized individuals}.
    \end{cases}
\end{equation*}

Figure \ref{fig:data_structure_schematic} provides a schematic of the data structure and highlights that the observed data, after sub-sampling of non-randomized individuals, have a monotone missing data pattern.

\subsection{Sampling properties}

In the nested trial design with sub-sampling of non-randomized individuals, we collect data on baseline covariates $X = (X_1, X_2)$; treatments, $A$; and outcomes, $Y$, from all randomized individuals, such that \[ \Pr[D = 1 | X, A, Y, S = 1] = 1. \] Furthermore, $X_1$ data are collected from all non-randomized individuals, but $X_2$ data are collected only from a subset. The sampling probability with which non-randomized individuals are selected for additional data collection depends only on $X_1$, that is, \[ \Pr[D = 1 | X, A, Y, S = 0] = \Pr[D = 1 | X_1, S = 0].\] Thus, the study design ensures that $D \indep (X, A, Y) | X_1, S$, which is a missing at random condition \cite{rubin1976inference}. Furthermore, by design, the sampling probability should be positive, \[\Pr[D = 1 | X_1 = x_1, S = 0] > 0,\] for all $x_1$ that have positive density among non-randomized individuals, $f_{X_1 | S}(x_1 | S = 0) > 0$. In the context of case-referrent studies, an approach similar to ours has been termed ``biased sampling'' \cite{weinberg1990design} or ``randomized recruitment'' \cite{weinberg1991randomized}.

For convenience, we define the following conditional sampling probability function:
\begin{equation*}
  \begin{split}
    c(X_1, S) \equiv \Pr[D = 1 | X_1, S] &= I(S = 1) + I(S = 0) \times \Pr[D = 1 | X_1, S = 0].
  \end{split}
\end{equation*}
A special case occurs when sub-sampling non-randomized individuals with probabilities that do not depend on baseline covariates, in which case the sampling function becomes  
\begin{equation*}
  \begin{split}
    c(S) \equiv \Pr[D = 1 | S] &= I(S = 1) + I(S = 0) \times c,
  \end{split}
\end{equation*}
where $c$ is a known constant, $ 0 < c < 1$.

\subsection{Causal quantities}

In order to define the causal contrasts of interest, let $Y^a$ denote the potential (counterfactual) outcome \cite{rubin1974, robins2000d} under intervention to set treatment to $a \in \mathcal A$.

We are interested in the potential outcome means in the target population of all trial-eligible individuals, $\E[Y^a]$. These potential outcome means are of inherent scientific interest and can also be used to identify average causal effects. For example, for $a, a^\prime \in \mathcal A$, the average treatment effect is $\E[Y^a - Y^{a^\prime}] = \E[Y^a ] - \E[Y^{a^\prime}]$.

\subsection{Identifiability conditions}

We assume that the following identifiability conditions hold for each $a \in \mathcal A$ \cite{dahabreh2018generalizing}: \emph{(I) Consistency of potential outcomes:} interventions are well-defined, so that if $A_i = a,$ then $Y^a_i = Y_i$. Implicit in this notation is that the offer to participate in the trial and trial participation itself do not have an effect on the outcome except through treatment assignment. \emph{(II) Mean exchangeability among trial participants:} $\E [ Y^a | X, S = 1, A = a] = \E [ Y^a | X, S = 1] $. This condition is expected to hold because of randomization (marginal or conditional on $X$). \emph{(III) Positivity of treatment assignment in the trial:} $\Pr[A =a | X = x, S = 1] > 0$ for each $x$ with positive density in the trial, $f_X(x | S = 1) > 0$. \emph{(IV) Mean generalizability (exchangeability over $S$):} $\E[ Y^a | X, S = 1] = \E[Y^a | X]$. Because $S$ is binary, this condition implies the mean transportability condition $\E[ Y^a | X, S = 1] = \E[Y^a | X, S = 0 ]$. \emph{(V) Positivity of trial participation:} $\Pr[S = 1 | X = x] > 0$, for each $x$ with positive density in the target population, $f_X(x) > 0$. 
% As will become clear in the next section, for some causal quantities of interest, it suffices that $\Pr[S = 1 | X = x] > 0$ for each $x$ with positive density in the non-randomized individuals, $f_X(x) > 0$.

Here, we have used $X$ generically to denote baseline covariates. It is possible however, that strict subsets of $X$ are adequate to satisfy the different exchangeability conditions. For example, in a marginally randomized trial the mean exchangeability among trial participants holds unconditionally. Furthermore, to focus on issues related to selective trial participation, we will assume complete adherence to the assigned treatment and no loss-to-follow-up.

%%%%%%%%%%%%%%%%%%%%%%%%%%%%%%%%%%%%%%%%%%%%%%%%%%%%%%%%%%%%%%%
\section{Identification}\label{sec:identification}
%%%%%%%%%%%%%%%%%%%%%%%%%%%%%%%%%%%%%%%%%%%%%%%%%%%%%%%%%%%%%%%

Under identifiability conditions \emph{(I)} through \emph{(V)}, the potential outcome mean under treatment $a$ in the target population, $\E[Y^a]$, can be expressed as a function of the full (observable) data,
\begin{equation}\label{eq:identification}
	\begin{split}
	\E[Y^a] 	&= \E \big[ \E[Y|X,S=1, A =a] \big] \\
            &= \int \E[Y|X = x,S=1, A =a] dF_X(x),
	\end{split}
\end{equation}	
where $F_X(x)$ is the cumulative distribution function of $X$ in the target population. 

Because $X = (X_1, X_2)$ is observed solely when $D = 1$, that is, among randomized individuals and sampled non-randomized individuals, whereas only $X_1$ is observed when $D = 0$, the above result cannot be directly applied to the observed data in nested trial designs with sub-sampling. In this design, however, as we show in Appendix \ref{appendix:A_identification},
\begin{equation}\label{eq:identification_subsampling}
	\begin{split}
		\E \big [ \E[ Y | X, S=1, A =a] \big] 	&= \E\left[ \dfrac{I(D = 1)}{c(X_1, S)}  \E [ Y | X, S = 1, A = a ] \right] \\
										&= \E\Big[ \E \big[ \E[Y | X, S = 1, A = a] | X_1, S, D = 1  \big]    \Big].
	\end{split}
\end{equation}
The above re-expression, together with the result in (\ref{eq:identification}), shows that $\E[Y^a]$ is identifiable in the nested trial design with sub-sampling of non-randomized individuals.

%%%%%%%%%%%%%%%%%%%%%%%%%%%%%%%%%%%%%%%%%%%%%%%%%%%%%%%%%%%%%%%
\section{Estimation and inference}
%%%%%%%%%%%%%%%%%%%%%%%%%%%%%%%%%%%%%%%%%%%%%%%%%%%%%%%%%%%%%%%

\subsection{Estimation}

We wish to estimate the functional $ \psi(a) = \E\Big[ \E \big[ \E[Y | X, S = 1, A = a] | X_1, S, D = 1  \big] \Big]$ under the semi-parametric model described by the identifiability conditions and sampling properties in Section \ref{sec:identification}. Using the efficient influence function \cite{bickel1993efficient} of $\psi(a)$ (see Appendix \ref{appendix:B_estimation_robustness} for details), we obtain the following one-step, in-sample estimator of $\psi(a)$,
\begin{equation}\label{eq:estimator}
\widehat \psi(a) = \dfrac{1}{n} \sum\limits_{i=1}^n \Biggl\{ \widehat b_a(X_{1i}, S_i) + \dfrac{I(D_i = 1)}{ c(X_{1i}, S_i)} \Big\{ \widehat g_a(X_i)  - \widehat b_a(X_{1i}, S_i) \Big\} + \dfrac{I(S_i = 1, A_i =a )}{\widehat p(X_i)  e_a(X_i)} \Big\{ Y_i  - \widehat g_a(X_i) \Big\}  \Biggl\},
\end{equation}
where $ c(X_{1}, S) = I(S=1) + I(S = 0) \times  \Pr[D = 1 | X_1, S = 0]$; $\widehat b_a(X_{1}, S)$ is an estimator for $\E \big[ \E[Y | X, S = 1, A = a] | X_1, S, D = 1  \big]$; $\widehat p(X)$ is an estimator for $\Pr[S = 1 | X]$; $ e_a(X) = \Pr[A =a | X, S = 1]$; and $\widehat g_a(X)$ is an estimator for $\E[Y | X, S = 1, A = a]$. Note that $\Pr[D = 1 | X_1, S = 0]$ and $\Pr[A =a | X, S = 1]$ are known by design, but they may also be estimated from the data. Estimating these known functions does not affect the large-sample behavior of the estimator \cite{robins1994estimation, robins1995semiparametric, lunceford2004, williamson2014variance}.

\vspace{0.15in}
\noindent
\emph{Estimating the probability of trial participation:} The computation of $\widehat \psi(a)$ requires the estimation of $\Pr[S = 1 | X]$, the population probability of trial participation, which is not the same as the probability of trial participation among sampled individuals, $\Pr[S = 1 | X , D = 1]$. Nevertheless, in the nested trial design with sub-sampling,
\begin{equation*}%\label{eq:id_prob_part}
\dfrac{\Pr[S = 1 | X]}{\Pr[S = 0 | X]} = \dfrac{\Pr[S = 1 | X, D = 1]}{\Pr[S = 0 | X, D = 1]} \times  c(X_1, S = 0),
\end{equation*}
and, clearly, the right-hand-side of the above expression is identifiable.% and can be estimated by fitting a regression model of $S$ on $X$ among sampled individuals ($D=1$), and then using the estimated odds of trial participation among the sampled, and the sampling probability, $c(X_{1}, S)$, to calculate the right hand side of (\ref{eq:id_prob_part}) and, thus, estimate $\widehat p(X)$. %Note also that, in (\ref{eq:estimator}), $\widehat p(X)$ is only needed among randomized individuals.

A straightforward estimation approach is to posit a parametric model for the population probability of trial participation, say, $\Pr[S = 1 | X] = p(X; \gamma)$ with finite dimensional parameter $\gamma$. We can estimate $\gamma$ by maximizing the pseudo-likelihood function
\begin{equation*}
\mathscr{L}(\gamma) = \prod\limits_{i=1}^n \big [ p(X_i; \gamma) \big ]^{S_i D_i} \big[ 1-  p(X_i; \gamma) \big]^{ (1 - S_i) D_i / c(X_{1i}, S_i)}.
\end{equation*}
Under reasonable technical conditions \cite{newey1994large, cosslett1981maximum}, the extremum estimator
\begin{equation*}
\widehat \gamma = \argmax_{\gamma \in \Gamma} \sum\limits_{i=1}^n \Bigg\{ S_i D_i \log p(X_i; \gamma) + \dfrac{(1 - S_i) D_i}{c(X_{1i}, S_i)} \log \big[ 1-  p(X_i; \gamma) \big] \Bigg\},
\end{equation*}
where $\Gamma$ is a compact parameter space, is a consistent estimator of $\gamma$, provided the population model is correctly specified. Thus, for example, using weighted regression of $S$ on $X$ among individuals with $D=1$ (i.e., randomized and sampled non-randomized individuals), with weights equal to $1 / c(X_1, S)$, we can obtain a consistent estimator of $\Pr[S = 1 | X]$, provided the parametric model $p(X; \gamma)$ is correctly specified \cite{manski1977estimation, cosslett1981maximum, scott1986fitting}.

\vspace{0.15in}
\noindent
\emph{Double robustness:} Suppose that $\widehat b_a(X_1, S)$, $\widehat g_a(X)$, and $\widehat p(X)$, have well-defined limiting values $b_a^*(X_1, S)$, $g_a^*(X)$, and $p^*(X)$, respectively. As we show in Appendix \ref{appendix:B_estimation_robustness}, $\widehat \psi(a)$ is \emph{doubly robust} in the following sense: $\widehat \psi(a)$ converges in probability to $\psi(a)$, that is, $\widehat \psi(a) \overset{p}{\longrightarrow} \psi(a)$, when \emph{either} $\widehat g_a(X) \overset{p}{\longrightarrow}  g_a^*(X) =  \E[Y|X, S = 1, A = a]$ \emph{or} $\widehat p(X) \overset{p}{\longrightarrow} p^*(X) = \Pr[S = 1 | X],$ but not necessarily both, and regardless of whether $b_a^*(X_1, S)$ is equal to $\E\big[\E [Y | X, S = 1, A =a] \big| X_1, S, D = 1 \big]$.

Last, we note that when there is no sub-sampling, that is when baseline covariate data on $X$ are collected from all non-randomized individuals, $c(X_{1}, S) =1$ and $D = 1$ in the entire sample, and the estimator becomes 
\begin{equation}\label{eq:estimator_nosub}
\widehat \psi_{\text{\tiny nosub}}(a) = \dfrac{1}{n} \sum\limits_{i=1}^n  \Biggl\{ \widehat g_a(X_i) + \dfrac{I(S_i = 1, A_i =a )}{\widehat p(X_i) e_a(X_i)} \Big\{ Y_i  - \widehat g_a(X_i) \Big\}  \Biggl\}.
\end{equation}
As we have shown before \cite{dahabreh2018generalizing}, this is the efficient estimator under the nested trial design without sub-sampling non-randomized individuals in the cohort.

\subsection{Inference}

In Appendix \ref{appendix:C_asymptotic_distribution} we derive the asymptotic distribution of $\widehat \psi(a)$ and thoroughly consider the impact of model misspecification. Here, we outline some key results. 

When $g_a(X)$ and $p(X)$ are consistently estimated using correctly specified models (and at sufficiently fast rate), $b_a(X_1, S)$ is estimated at $\sqrt{n}$-rate (even with a misspecified model), and regardless of whether $c(X_1, S)$ or $e_a(X)$ are estimated using correctly specified models (and at sufficiently fast rate) or are known, the estimator in (\ref{eq:estimator}) is \emph{locally efficient,} in the sense that it attains the variance bound under the semi-parametric model defined by the identifiability conditions, the sampling properties, and the additional model restrictions. 

In Appendix \ref{appendix:D_asymptotic_efficiency} we show that, under correct model specification, the estimator has asymptotic variance 
\begin{equation}
	\mbox{AVar}_1 = \E \left[ \dfrac{v_a(X)}{p(X) e_{a}(X)}  \right] + \mbox{Var} \big[ g_{a}(X) \big] + \E\left[ \dfrac{1 - c(X_1, S)}{c(X_1, S)} \Big\{ g_{a}(X) - b_{a}(X_1, S)      \Big\}^2 \right],
\end{equation}
where  $v_a(X) = \mbox{Var}[Y | X, S = 1, A = a]$ and all quantities are evaluated at the true law. The above result implies that, when models are correctly specified, the asymptotic variance of the efficient estimator for the nested trial design  sub-sampling in (\ref{eq:estimator}) is greater than or equal to the asymptotic variance of the efficient estimator for the nested trial design without sub-sampling non-randomized individuals in (\ref{eq:estimator_nosub}); see Appendix \ref{appendix:D_asymptotic_efficiency} for details. 

To construct Wald-style confidence intervals for $\psi(a)$, when using parametric models, we can easily obtain the sandwich estimator \cite{stefanski2002} of the sampling variance of the estimator in (\ref{eq:estimator}) by solving the appropriate estimating equation, jointly with the estimating equations for the parameters of the working models for $b_a(X_1, S)$, $g_a(X)$, and $p(X)$ \cite{lunceford2004}. Alternatively, we can use the non-parametric bootstrap \cite{efron1994introduction}.  The results presented in Appendix \ref{appendix:C_asymptotic_distribution} ensure that the bootstrap-based standard error estimator is valid \cite{praestgaard1993exchangeably}. 

A less computationally demanding large-sample $(1-\alpha)$\% confidence interval for $\psi(a)$ can be obtained \cite{vanderLaan2003} as $$ \widehat \psi(a) \pm  z_{1-\alpha/2} \widehat{SE}\Big[\widehat \psi(a)\Big],$$ where $z_{1-\alpha/2}$ is the $(1-\alpha/2)$th quantile of the standard normal distribution and $\widehat{SE}\Big[\widehat \psi(a)\Big]$ is given by $$\widehat{SE}\Big[\widehat \psi(a) \Big] =  \dfrac{1}{n}  \sqrt{\sum\limits_{i=1}^{n} \widehat{IC}_i^2 },$$ and $$ \widehat{IC}_i  =  \widehat b_a(X_{1i},S_i) + \dfrac{I(D_i = 1)}{ c(X_{1i}, S_i)} \Big\{ \widehat g_a(X_i)  - \widehat b_a(X_{1i},S_i) \Big\} + \dfrac{I(S_i = 1, A_i =a )}{\widehat p(X_i)  e_a(X_i)} \Big\{ Y_i  - \widehat g_a(X_i) \Big\}.$$

%%%%%%%%%%%%%%%%%%%%%%%%%%%%%%%%%%%%%%
\section{Simulation study}
%%%%%%%%%%%%%%%%%%%%%%%%%%%%%%%%%%%%%%

\subsection{Methods}

Building on our earlier work \cite{dahabreh2018generalizing}, we conducted a simulation study to examine the finite-sample performance of the estimator in (\ref{eq:estimator}) when sub-sampling non-randomized individuals, and compare it against the estimator in (\ref{eq:estimator_nosub}), without sub-sampling. We simulated scenarios using trials with an average sample size of 1000 individuals nested in cohorts of 2000, 5000, or 10,000 individuals and scenarios using trials with an average sample size of 2000 individuals nested in cohorts of 5000, 10,000, or 20,000 individuals. Appendix \ref{appendix:E_simulation_details} provides details about the scenarios we considered. 

\vspace{0.1in}
\noindent
\emph{Nested trial data generation:} We generated data for three baseline covariates, $Z_j, j = 1, 2,3$. Thus, in the notation of the previous section, $(Z_1, Z_2, Z_3) = X.$ For $Z_1$, we considered both continuous and binary distributions; for the continuous case, $Z_1 \sim \mathcal{N} (0,1)$; for the binary case, $Z_1 \sim \text{Bernoulli} (0.5) $. For $j=2,3$, we used $Z_j \sim \mathcal{N} (0,1)$, $i=1, ..., n$. We generated ``selection'' into the trial using a logistic linear model for the trial participation indicator, $$S \sim \text{Bernoulli} ( \Pr[S = 1 | Z ] ) \mbox{ with } \Pr[S = 1 | Z ] = \dfrac{\text{exp}(\gamma Z^T) }{1+ \text{exp}(\gamma Z^T)},$$ $Z = (1, Z_1, \ldots, Z_3)$, $\gamma = (\gamma_0, 1, 1,1 )$, and intercept $\gamma_0$ chosen for each $n$ such that it resulted in randomized trials with the desired average sample size (see Appendix \ref{appendix:E_simulation_details}). We determined the $\gamma_0$ for each scenario using the numerical methods described in \cite{Austin2013}. We generated an indicator of unconditionally randomized treatment assignment, $A$, among randomized individuals using a Bernoulli distribution with parameter  $ \Pr [A = 1 | S = 1] = 0.5$. We then generated continuous outcomes using linear potential outcome models with normally distributed errors: $Y^a =   \theta^a Z^T + \epsilon^a, \mbox{ for } a \in \{0,1\},$ where $\theta^{a} = (\theta^a_0, \ldots , \theta_3^a), a \in \{0,1\}$. We set $\theta^{0} = (1,1,1,1)$ and $\theta^1 = (0, 0, 0, 1)$ (i.e., effect modification by both $Z_1$ and $Z_2$). In all simulations, $\epsilon^a$ had a standard normal distribution for $a=0,1$. We generated observed outcomes as $ Y = A Y^1 + (1 - A) Y^0.$ 

\vspace{0.1in}
\noindent
\emph{Sub-sampling of non-randomized individuals:} We assumed that $Z_1$ was measured on all cohort members, but $Z_2$ and $Z_3$ were only measured on randomized individuals and sampled non-randomized individuals. In the notation of the previous section, $Z_1 = X_1$ and $(Z_2, Z_3) = X_2$. After generating the cohort data, we sub-sampled non-randomized observations ($S=0$) using a logistic linear model, $$D \big| Z, S = 0 \sim \text{Bernoulli} ( \Pr[D = 1 | Z_1 , S = 0] ) \mbox{ with } \Pr[D = 1 | Z_1 , S = 0] = \dfrac{\text{exp}(\zeta_0 + Z_1) }{1+ \text{exp}(\zeta_0 + Z_1)},$$ with the intercept, $\zeta_0$, chosen for each cohort sample size, $Z_1$ distribution (continuous or binary), and marginal probability of trial participation, such that it resulted in marginal sampling probabilities of non-randomized individuals, $\Pr[D = 1 | S = 0]$, ranging from 0.1 to 0.9, in steps of 0.1 (see Appendix Tables \ref{table:sample_size_scenarios_cont} and \ref{table:sample_size_scenarios_binary} for details). As for selective trial participation, we determined the $\zeta_0$ for each scenario using the numerical methods described in \cite{Austin2013}. We also considered a case where the sampling probability did not depend on baseline covariates, but instead a simple random sample of the non-randomized patients was taken. For all randomized individuals, we set $D = 1$ in all simulations.

\vspace{0.1in}
\noindent
\emph{Comparisons and performance measures:} In each simulated dataset, we applied estimator (\ref{eq:estimator}) to the sub-sampled data and estimator (\ref{eq:estimator_nosub}) to non-sub-sampled data. For each estimator, we estimated bias, variance, and mean squared error over 10,000 runs for each scenario. 

In the simulations, the working models for $g_a(Z) = \E[Y|Z, S = 1, A = a], e_a(Z) = \Pr[A = a | Z, S = 1]$, and $p(Z) = \Pr[S =1 | Z]$ were correctly specified, in the sense that the true model was included within the class of models under consideration. Specifically, models for the probability of participation in the trial and models for the probability of treatment included all main covariate effects; outcome models were fitted separately to each treatment group (i.e., allowed for all possible treatment by covariate interactions over all covariates). For $b_a(Z_1, S) = \E \big[ \E[Y | Z, S = 1, A = a]\big| Z_1, S, D = 1\big]$ we used the linear model $b_a(Z_1, S) = \xi_0 + \xi_1 Z_1 + \xi_2 S + \xi_3 Z_1 \times S$. We estimated $c(Z_1, S = 0) = \Pr[D = 1 | Z_1, S = 0]$ using a correctly specified logistic model with $Z_1$ as the only covariate, fit among non-randomized individuals; and we set $c(Z_1, S = 1) = \Pr[D = 1 | Z_1, S = 1] = 1$.

\subsection{Results}

Complete results from the simulation study are presented in Appendix Tables \ref{table:BIAS_continuous_results} through \ref{table:VARIANCE_binarySRS_results}. In all simulations, when all models were correctly specified, estimator (\ref{eq:estimator}), which uses the sub-sampled data, was nearly unbiased, for marginal sampling probabilities of non-randomized individuals ranging from 0.1 to 0.9, despite the presence of strong selection on baseline covariates and strong effect modification. As expected based on prior work \cite{dahabreh2018generalizing}, the estimator in (\ref{eq:estimator_nosub}), which uses the non-sub-sampled data, was also nearly unbiased.

Results for the sampling variance of the estimators for scenarios with an average trial sample size of 1000 observations are graphed in Figure \ref{fig:simulation_variance1000}; results for scenarios with an average trial sample size of 2000 observations are graphed in Appendix Figure \ref{fig:simulation_variance2000}. For both trial sample sizes and regardless of the cohort sample size, with increasing marginal sampling probability of non-randomized individuals, the variance of the estimator in (\ref{eq:estimator}) approached the variance of the estimator in (\ref{eq:estimator_nosub}) (the latter applied only to data without sub-sampling). In this simulation, the sampling variances of the two estimators were quite similar once the marginal sampling probability was greater than 0.3. 

Results were similar when the sampling probabilities depended on $Z_1$, a baseline covariate that was both an effect modifier and a strong predictor of trial participation (both when $Z_1$ was continuous and discrete), and when the sub-sampling did not depend on $Z_1$ (i.e., in the case of simple random sampling). Of note, in all simulation scenarios the sampling variance was larger with increasing cohort size, because, holding the trial sample size and the selection and sub-sampling mechanisms constant, the difference in the covariate distribution between randomized and non-randomized observations increases as the cohort sample size increases (i.e., as the marginal probability of trial participation decreases).

%%%%%%%%%%%%%%%%%%%%%%%%%%%%%%%%%%%%%%%%%%%%%%%%%%%%%%%%%%%%%%%%%%%%%%%%%%%%%%
\section{The Coronary Artery Surgery Study (CASS)} \label{section_example}
%%%%%%%%%%%%%%%%%%%%%%%%%%%%%%%%%%%%%%%%%%%%%%%%%%%%%%%%%%%%%%%%%%%%%%%%%%%%%%

\subsection{CASS design and data}

CASS was a comprehensive cohort study that compared coronary artery bypass grafting surgery plus medical therapy (henceforth, ``surgery'') versus medical therapy alone for individuals with chronic coronary artery disease; details about the design of CASS are available elsewhere \cite{william1983, investigators1984}. In brief, individuals undergoing angiography in 11 institutions were screened for eligibility and the 2099 trial-eligible individuals who met the study criteria were either randomized to surgery or medical therapy (780 individuals), or included in an observational study (1319 individuals). We excluded 6 individuals for consistency with prior CASS analyses and in accordance with CASS data release notes; in total we used data from 2093 individuals (778 randomized; 1315 non-randomized). Baseline covariates were collected from randomized and non-randomized individuals in an identical manner. No randomized individuals were lost to follow-up in the first 10 years of the study; we did not use information on adherence among randomized individuals, in effect assuming that the non-adherence would be similar among all eligible individuals. 

In \cite{dahabreh2018generalizing} we used these data to illustrate generalizability methods for nested trial designs \emph{without} sub-sampling. Here, we build on that work to illustrate the use of methods that are appropriate when the full covariate data is only obtained from a subset of non-randomized individuals. To do so, we \emph{emulated} the sub-sampling design under a variety of scenarios. We assumed that clinical covariates were measured on all cohort members (both randomized and non-randomized), but laboratory covariates were only measured on randomized individuals and sampled non-randomized individuals. We sub-sampled the non-randomized individuals using (1) covariate-dependent sampling of non-randomized individuals, where sampling depended on past history of myocardial infarction, such that individuals without a history of infarction had double the probability of being sampled compared to individuals with such history, \emph{and} marginal probability of sampling ranging from 0.1 to 0.7, in steps of 0.1 (sampling probabilities of 0.8 and 0.9 were not possible to implement while preserving the aforementioned relationship between the sampling probabilities of individuals with and without history of infarction because they corresponded to probabilities greater than 1 for one of these subgroups); and (2) simple random sampling of non-randomized individuals with probabilities that ranged from 0.1 to 0.9, in steps of 0.1.

\subsection{Statistical analysis}

\emph{Estimands and estimators:} We estimated the 10-year mortality risk under surgery and medical therapy, and the risk difference comparing the treatments for the target population of all trial-eligible individuals. We applied the estimator in (\ref{eq:estimator}) to sub-sampled data and compared it against the estimator in (\ref{eq:estimator_nosub}) applied to the original, non-sub-sampled, data. 

\emph{Working models:} We fit logistic regression models for the probability of participation in the trial among sub-sampled individuals, the probability of treatment among randomized individuals, and the probability of the outcome (in each treatment arm), conditional on \emph{clinical covariates} (age, severity of angina, history of previous myocardial infarction) and \emph{laboratory covariates} (percent obstruction of the proximal left anterior descending artery, left ventricular wall motion score, number of diseased vessels, and ejection fraction). We chose these variables based on a previous analysis of the same data \cite{olschewski1992}.

We considered three ways of obtaining the probability of sampling a non-randomized individual for use in (\ref{eq:estimator}): (1) use the ``design-based'' (known) sampling probabilities; (2) estimate the sampling probabilities, that is, use the empirical proportion in simple random sampling scenarios, or estimate the probability using a logistic regression model with history of infarction as the only covariates in covariate-dependent sampling scenarios; and (3) estimate the sampling probabilities using a logistic regression model that inlcuded all clinical covariates (even if not used to determine the sampling probabilities by design).

\emph{Missing baseline covariate data:} Of the 2093 trial-eligible individuals, 1686 had complete data on all baseline covariates (731 randomized, 368 in the surgery group and 363 in the medical therapy group; 955 non-randomized). In \cite{dahabreh2018generalizing} we undertook extensive missing data analyses under a missing at random assumption, which produced results very similar to those of the complete case analyses. For simplicity, here, we only report analyses restricted to individuals with complete data. 

\emph{Inference:} For all analyses, we used bootstrap resampling (with 10,000 samples) to estimate standard errors.

\subsection{Results}

We summarize comparisons as ratios of the estimated standard error of the estimator in (\ref{eq:estimator}), for each sub-sampling scenario and for each method of obtaining the sub-sampling probability, divided by the standard error of the estimator in (\ref{eq:estimator_nosub}) applied to the original CASS data. In general, except when the marginal sampling probability was less than 0.2, the standard errors were very similar, suggesting that sub-sampling non-randomized individuals in this example would not have adversely affected precision. Figure \ref{fig:cass_srs} and Appendix Figure \ref{fig:cass_dependent} summarize results from analyses under covariate-dependent and simple random sampling, respectively.

\section{Discussion}

We provide identification and estimation results for nested trial designs with sub-sampling of non-randomized individuals. These designs aim to support the generalization of causal inferences from randomized trials to the target population of all trial-eligible individuals, while improving research economy by limiting the collection of baseline covariates (in particular, covariates that are expensive to collect) to a subset of non-randomized individuals.

Nested trial designs will be increasingly implemented in conjunction with pragmatic randomized trials \cite{ford2016} because data from these trials can be linked with routinely collected (``real-world'') data (e.g., insurance claims or electronic health records). This linkage creates datasets that merge the trial data with routinely collected observational data, with a common set of baseline covariates available both from randomized and non-randomized trial-eligible individuals. In applications, the generalization of inferences from the randomized individuals to the target population of all trial-eligible individuals will often require information on covariates that are collected in the randomized trial but are not readily available in the routinely collected observational data (e.g., specialized imaging or laboratory tests). When additional data collection from non-randomized individuals is necessary, sub-sampling of non-randomized individuals for additional data collection, combined with efficient statistical estimation methods, can support inferences that are almost as precise as those possible by collecting data from all non-randomized individuals, but at a fraction of the cost. 

In our simulations and the CASS re-analysis, we found that the performance of our sub-sampling estimator quickly approached that of the efficient estimator under no sub-sampling as the marginal probability of sampling non-randomized individuals increased. We conjecture that this pattern should be expected in most cases because the main contribution of the sampled non-randomized individuals is in the estimation of the conditional covariate distribution, $F_{X|S}( x | S = 0)$; this conditional distribution enters the identification result in (\ref{eq:identification}), because $F_X(x) = \sum_{s} F_{X|S}( x | S = s)\Pr[S = s]$. Provided the total cohort sample size is fairly large and the sampling mechanism is reasonably chosen (e.g., sampling probabilities are away from 0), $F_{X|S}( x | S = 0)$ is estimated well even at low marginal sampling probabilities. Thus, in practical applications of nested trial designs with sub-sampling of non-randomized individuals, it will often be wise to focus resources on better estimating $\E[Y | X, S = 1, A = a]$ and $F_{X|S}( x | S = 1)$, by increasing the sample size of the randomized trial. Such focus will also serve to strengthen the trial-specific inferences that are the usual reason for conducting the trial in the first place \cite{dahabreh2018randomization}. 

In summary, we have described nested trial designs with sub-sampling of non-randomized individuals for generalizing causal inferences from randomized trials to the target population of all trial-eligible individuals. Our study of efficient and robust estimation methods for these designs suggests that sub-sampling can improve research economy without severely affecting precision.

\section{Acknowledgments}

We thank Professor Stephen Cole and the participants of the Causal Inference Research Group monthly meeting on March 1, 2019, at the University of North Carolina - Chapel Hill, for helpful comments on an earlier version of this work.

This work was supported in part by Patient-Centered Outcomes Research Institute (PCORI) awards ME-1306-03758 and ME-1502-27794 (Dahabreh); National Institutes of Health (NIH) grant R37 AI102634 (Hern\'an); Agency for Healthcare Research and Quality (AHRQ) National Research Service Award T32AGHS00001 (Robertson); and NIH grants DP2DA046856 and U54GM115677 (Buchanan). The content is solely the responsibility of the authors and does not necessarily represent the official views of PCORI, its Board of Governors, the PCORI Methodology Committee, the NIH, or AHRQ. The data analyses in our paper used CASS research materials obtained from the NHLBI Biologic Specimen and Data Repository Information Coordinating Center. This paper does not necessarily reflect the opinions or views of the CASS or the NHLBI.

\clearpage
\section{Figures}

\begin{figure}[!htbp]
\caption{Schematic of the observed data structure for nested trial designs with sub-sampling of non-randomized individuals. Gray shading indicates that the variable is not measured on some of the individuals in the cohort.}\label{fig:data_structure_schematic}
\centering
  \includegraphics[scale = 1.3]{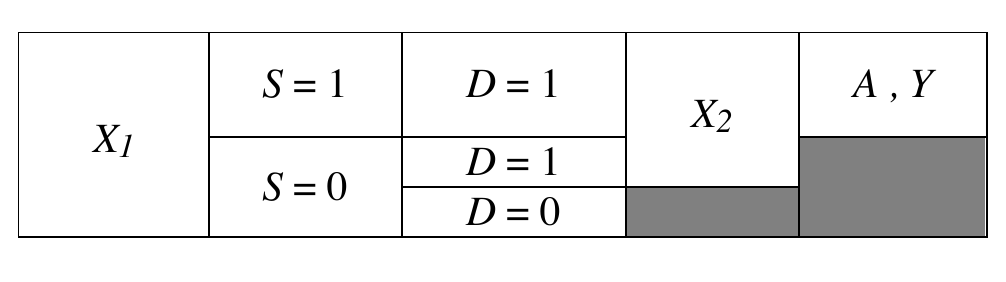}.
\end{figure}

\clearpage
\begin{figure}[!htbp]
\caption{Simulation results for the sampling variance of estimators for $\psi(a), a = 0,1$ and $\psi(1) - \psi(0)$, with average trial sample size of 1000 observations. Results in each panel are shown for different data generating mechanisms (binary or continuous $Z_1$) and sampling mechanisms (dependent on $Z_1$ or simple random sampling, SRS). In all panels, results are shown for $\widehat \psi(a)$ under marginal sampling probabilities ranging from 0.1 to 0.9, in steps of 0.1 (black markers); and for $\widehat \psi_{\text{\tiny nosub}}(a)$ under no sub-sampling (white markers). In each panel, results are shown for cohort sample sizes of 2000 (circles), 5000 (triangles), and 10,000 (squares) individuals.}\label{fig:simulation_variance1000}
\centering
  \includegraphics[scale = 1.8]{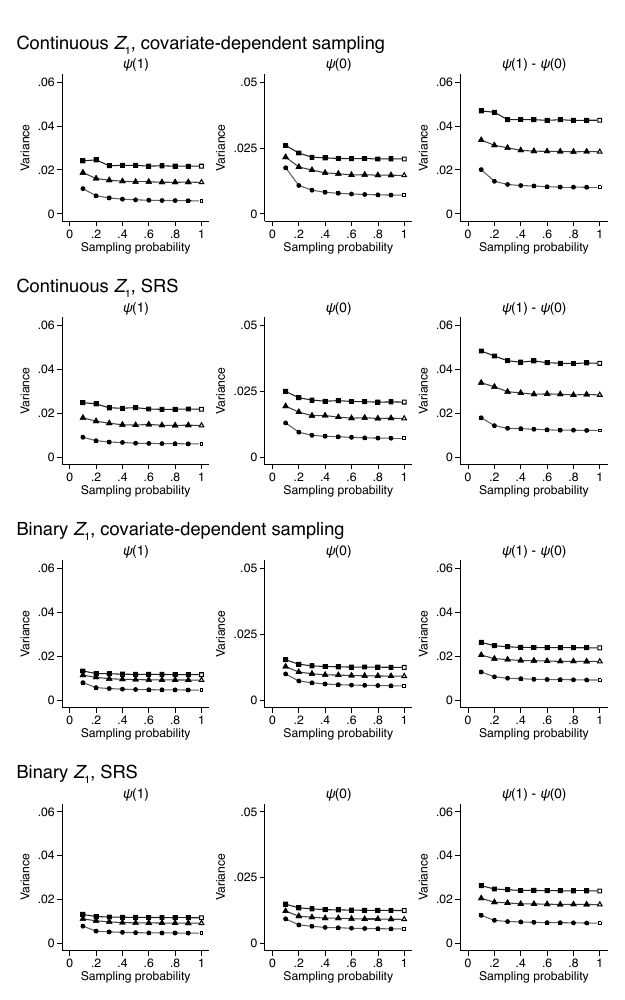}
\end{figure}

\clearpage
\begin{figure}[!htbp]
	\caption{CASS analysis results comparing the estimated standard errors of the estimator in (\ref{eq:estimator}) for $\psi(a), a = 0,1$ and $\psi(1) - \psi(0)$, under covariate dependent sampling with sampling probabilities of non-randomized individuals that depended on history of myocardial infarction, against the estimator in (\ref{eq:estimator_nosub}); see main text for details. In all panels, standard error ratios are shown for marginal sampling probabilities ranging from 0.1 to 0.7, in steps of 0.1 (in this analysis marginal probabilities of 0.8 or 0.9 were infeasible under the chosen covariate dependence relationship).}\label{fig:cass_dependent}
	\centering
	\includegraphics[scale = 1.8]{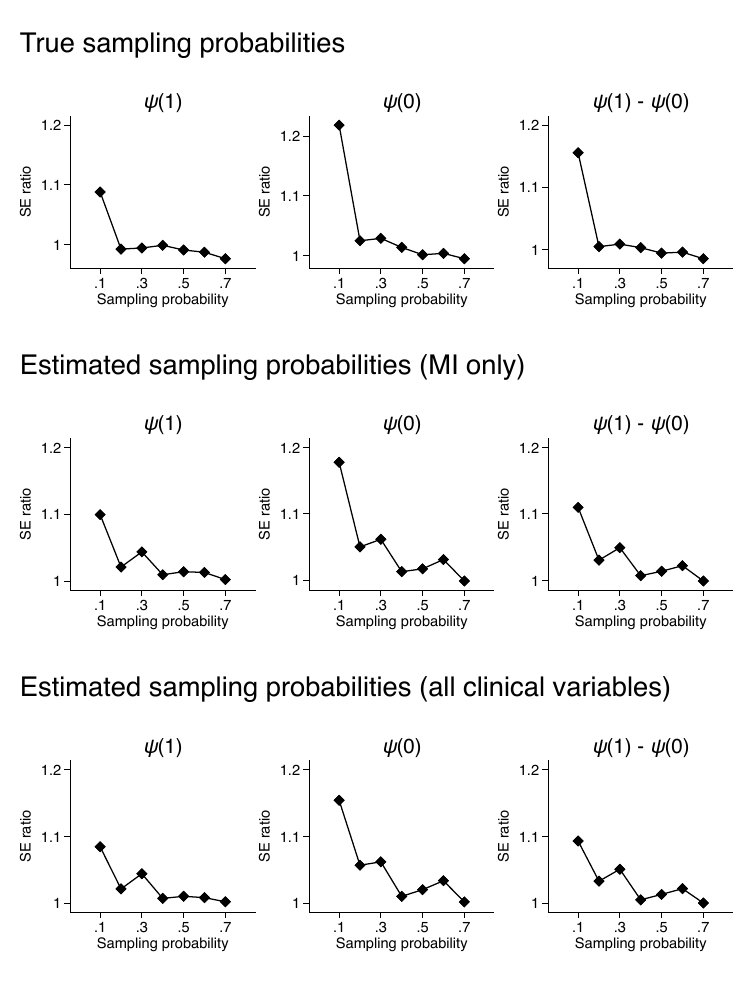}
\end{figure}

%%%%%%%%%%%%%%%%%%%%%%%%%%%%%%%%%%%%%%%%%%%%%%%%%%%%%%%%%%%%%%%%%%%%%%%%%%%%%%
% BIBLIOGRAPHY
%%%%%%%%%%%%%%%%%%%%%%%%%%%%%%%%%%%%%%%%%%%%%%%%%%%%%%%%%%%%%%%%%%%%%%%%%%%%%%
\clearpage
\bibliographystyle{unsrt}
\bibliography{subsampling_manuscript}

%%%%%%%%%%%%%%%%%%%%%%%%%%%%%%%%%%%%%%%%%%%%%%%%%%%%%%%%%%%%%%%%%%%%%%%%%%%%%%
% VERSIONING 
%%%%%%%%%%%%%%%%%%%%%%%%%%%%%%%%%%%%%%%%%%%%%%%%%%%%%%%%%%%%%%%%%%%%%%%%%%%%%%

\ddmmyyyydate %redefine \today format
\newtimeformat{24h60m60s}{\twodigit{\THEHOUR}.\twodigit{\THEMINUTE}.32}
\settimeformat{24h60m60s}
\begin{center}
\vspace{\fill}\ \newline
\textcolor{black}{{\tiny $ $transportability\_subsampling, $ $ }
{\tiny $ $Date: \today~~ \currenttime $ $ }
{\tiny $ $Revision: \paperversionmajor.\paperversionminor $ $ }}
\end{center}
%%
%%

%%%%%%%%%%%%%%%%%%%%%%%%%%%%%%%%%%%%%%
\clearpage
\appendix
%%%%%%%%%%%%%%%%%%%%%%%%%%%%%%%%%%%%%%

%%%%%%%%%%%%%%%%%%%%%%%%%%%%%%%%%%%%%%
\section{Identification}\label{appendix:A_identification}
%%%%%%%%%%%%%%%%%%%%%%%%%%%%%%%%%%%%%%

\begin{proposition}
Under the identifiability conditions I through V listed in the main text, 
\begin{equation*}
	\begin{split}
	\E[Y^a] = \E\big [\E[Y|X,S=1, A =a] \big].
	\end{split}
\end{equation*}
\end{proposition}

\begin{proof}
	\begin{equation*}
		\begin{split}
		\E[Y^a] &= \E \big[  \E [Y^a | X] \big] \\
				&= \E \big[  \E [Y^a | X, S = 1] \big] \\
				&= \E \big[  \E [Y^a | X, S = 1, A = a] \big] \\
				&= \E \big[  \E [ Y | X, S = 1, A = a] \big].
		\end{split}
	\end{equation*}
\end{proof}

\begin{proposition}
In the nested trial design with sub-sampling of non-randomized individuals, where $ c(X_1, S) \equiv \Pr[D = 1 | X_1, S ]$, and, by design, $\Pr[D = 1 | X, A, Y, S] = \Pr[D = 1 | X_1, S ]$, 
\begin{equation*}
	\begin{split}
	\E\big [\E[Y|X,S=1, A =a] \big]  &= \E\left[ \dfrac{I(D = 1)}{c(X_1, S)}  \E [ Y | X, S = 1, A = a ] \right] \\
	&= \E\Big[ \E \big[ \E[Y | X, S = 1, A = a] | X_1, S, D = 1  \big]    \Big].
	\end{split}
\end{equation*}
\end{proposition}

\begin{proof}
For the first equality, 
\begin{equation*}
	\begin{split}
		\E\left[ \dfrac{I(D = 1)}{c(X_1, S)} \E [ Y | X, S = 1, A = a ] \right] &= \E\Bigg[  \E\left[ \dfrac{I(D = 1)}{c(X_1, S)}  \E [ Y | X, S = 1, A = a ]  \Big | X, A, Y, S  \right]   \Bigg]  \\
		&= \E\left[  \dfrac{\E [ Y | X, S = 1, A = a ]}{c(X_1, S)}   \E[ I(D = 1) | X, A, Y, S]  \right]  \\
		&= \E\big[ \E [ Y | X, S = 1, A = a ] \big],
	\end{split}
\end{equation*}
where the last equality follows from the sampling properties of nested-trial design, using $\Pr[D = 1 | X, A, Y, S] = \Pr[D = 1 | X_1, S] = c(X_1, S)$.

For the second equality,
\begin{equation*}
	\begin{split}
		\E\Big[ \E \big[ \E[Y | X, S = 1, A = a] | X_1, S, D = 1 \big] \Big] &= \E\Big[ \E \big[ \E[Y | X, S = 1, A = a] | X_1, S \big]    \Big]  \\
		&= \E\big[ \E [ Y | X, S = 1, A = a ] \big]
	\end{split}
\end{equation*}
where the first equality follows from the fact that, by design, $D \mathlarger{\mathlarger{\indep}} (X, A, Y) | X_1, S$.
\end{proof}

\begin{remark}
Our derivation for Proposition 2 only uses the sampling properties of the nested trial design with sub-sampling of non-randomized individuals and does not use any of the structural identifiability conditions. Thus, the result holds even when $\psi(a)$ does not have any causal interpretation.
\end{remark}

%%%%%%%%%%%%%%%%%%%%%%%%%%%%%%%%%%%%%%
\clearpage
\section{Estimation and double robustness}\label{appendix:B_estimation_robustness}
%%%%%%%%%%%%%%%%%%%%%%%%%%%%%%%%%%%%%%
\setcounter{equation}{0}
\renewcommand{\theequation}{B.\arabic{equation}}

%%%%%%%%%%%%%%%%%%%%%%%%%%%%%%%%%%%%%%
\subsection{Influence function}\label{subsec:influence _function}
%%%%%%%%%%%%%%%%%%%%%%%%%%%%%%%%%%%%%%

The influence function of $\psi(a)$ is 
\begin{equation*}
\Psi_0^1(a) =  b_{a}(X_{1}, S) + \dfrac{I(D = 1)}{c(X_{1}, S)} \Big\{  g_{a}(X)  -  b_{a}(X_{1}, S) \Big\} + \dfrac{I(S = 1, A =a )}{ p(X)  e_{a}(X)} \Big\{ Y - g_{a}(X) \Big\} - \psi(a),
\end{equation*}
where 
\begin{equation*}
  \begin{split}
    b_{a}(X_{1}, S) &= \E \big[ \E[Y | X, S = 1, A = a] | X_1, S, D = 1  \big], \\
    c(X_1, S) &= \Pr[D = 1 | X_1, S],\\
    g_{a}(X) &= \E[Y | X, S = 1, A = a], \\
    p(X) &= {\Pr}[S = 1 | X], \mbox{ and }\\
    e_{a}(X) &= {\Pr}[A = a | X, S = 1],
  \end{split}
\end{equation*}
and all quantities are evaluated at the ``true'' law.

%%%%%%%%%%%%%%%%%%%%%%%%%%%%%%%%%%%%%%
\subsection{Connection with two-stage designs}
%%%%%%%%%%%%%%%%%%%%%%%%%%%%%%%%%%%%%%

Applying the theory for two-stage designs from \cite{rose2011targeted}, we obtain the following expression for the influence function for nested trials with sub-sampling of non-randomized individuals:
\begin{equation}\label{eq:projection_if}
  \begin{split}
   \widetilde \Psi_0^1(a) = \dfrac{I(D = 1)}{c(X_1, S)} \Psi^{1*}_{0}(a) + \left\{ 1 - \dfrac{I(D = 1)}{c(X_1, S)} \right\} \E\big[ \Psi^{1*}_{0}(a) \big | X_1, S, D = 1 \big],
  \end{split}
\end{equation}
where $\Psi^{1*}_{0}(a)$ is the influence function under the nested trial design without sub-sampling (census) of non-randomized individuals \cite{dahabreh2018generalizing},
\begin{equation}\label{eq:IF_no_subsampling}
  \Psi^{1*}_{0}(a) = \dfrac{I(S = 1, A = a)}{p(X) e_{a}(X)} \Big\{ Y - g_{a}(X) \Big\} + g_{a}(X) - \psi(a),
\end{equation}
and, as before, all quantities are evaluated at the ``true'' law. We will now show that the above result is in agreement with our result in Section \ref{subsec:influence _function}.

\begin{proposition}
In the nested trial design with sub-sampling of non-randomized individuals, 
  \begin{equation}\label{eq:relationship_of_IFs}
    \Psi_0^1(a) =  \widetilde \Psi_0^1(a).
\end{equation}
\end{proposition}

\begin{proof}
We begin work on the right-hand-side of equation (\ref{eq:relationship_of_IFs}), plugging the no-sub-sampling influence function from (\ref{eq:IF_no_subsampling}) into equation (\ref{eq:projection_if}):
\begin{equation*}
  \begin{split}
\widetilde \Psi_0^1(a) &= \dfrac{I(D = 1)}{c(X_1, S)} \left\{ \dfrac{I(S = 1, A = a)}{p(X) e_{a}(X)} \Big\{ Y - g_{a}(X) \Big\} + g_{a}(X) - \psi(a) \right\} \\
&\quad+ \left\{ 1 - \dfrac{I(D = 1)}{c(X_1, S)} \right\} \E\left[ \dfrac{I(S = 1, A = a)}{p(X) e_{a}(X)} \Big\{ Y - g_{a}(X) \Big\} + g_{a}(X) - \psi(a) \Bigg | X_1, S, D = 1 \right]
  \end{split}
\end{equation*}

Noting that if $ S=1 $, then $c(X_1, S) = 1$ and $D = 1$, we obtain
\begin{equation*}
  \begin{split}
\widetilde \Psi_0^1(a) &=  \dfrac{I(S = 1, A = a)}{p(X) e_{a}(X)} \Big\{ Y - g_{a}(X) \Big\} +   \dfrac{I(D = 1)}{c(X_1, S)} \Big\{ g_{a}(X) - \psi(a) \Big\} \\
&\quad+ \left\{ 1 - \dfrac{I(D = 1)}{c(X_1, S)} \right\} \E\left[ \dfrac{I(S = 1, A = a)}{p(X) e_{a}(X)} \Big\{ Y - g_{a}(X) \Big\} + g_{a}(X) - \psi(a) \Bigg | X_1, S, D = 1 \right].
  \end{split}
\end{equation*}

Next, we note that,
\begin{equation*}
  \begin{split}
&\E\left[\dfrac{I(S = 1, A = a)}{p(X) e_{a}(X)} \Big\{ Y - g_{a}(X) \Big\} + g_{a}(X) - \psi(a) \Bigg | X_1, S, D = 1 \right] \\
&\quad\quad\quad=\E\left[\rule{0cm}{0.9cm}  \E\left[ \dfrac{I(S = 1, A = a)}{p(X) e_{a}(X)} \Big\{ Y - g_{a}(X) \Big\}  \Bigg | X, A, S, D = 1 \right] X_1, S, D = 1  \right] \\
&\quad\quad\quad\quad\quad\quad\quad\quad\quad+ \E\left[  g_{a}(X) - \psi(a) \big| X_1, S, D = 1  \right] \\
&\quad\quad\quad= \E\big[  g_{a}(X) - \psi(a) \big | X_1, S, D = 1 \big]  \\
&\quad\quad\quad= b_{a}(X_1, S) - \psi(a).
  \end{split}
\end{equation*}

Using the above result, we see that
\begin{equation*}
  \begin{split}
\widetilde \Psi_0^1(a) &=  \dfrac{I(S = 1, A = a)}{p(X) e_{a}(X)} \Big\{ Y - g_{a}(X) \Big\} +   \dfrac{I(D = 1)}{c(X_1, S)} \Big\{ g_{a}(X) - \psi(a) \Big\} \\
&\quad\quad\quad+ \left\{ 1 - \dfrac{I(D = 1)}{c(X_1, S)} \right\} \Big\{ b_{a}(X_1, S) - \psi(a) \Big\} \\
&=  b_{a}(X_1, S) +   \dfrac{I(D = 1)}{c(X_1, S)} \Big\{ g_{a}(X) -  b_{a}(X_1, S) \Big\} + \dfrac{I(S = 1, A = a)}{p(X) e_{a}(X)} \Big\{ Y - g_{a}(X) \Big\} - \psi(a) \\
&=  \Psi_0^1(a),
  \end{split},
\end{equation*}
which completes the proof.
\end{proof}

%%%%%%%%%%%%%%%%%%%%%%%%%%%%%%%%%%%%%%
\subsection{One-step in-sample estimator}\label{subsec:appendix_estiamtor}
%%%%%%%%%%%%%%%%%%%%%%%%%%%%%%%%%%%%%%

The influence function in section \ref{subsec:influence _function} suggests the following in-sample one-step estimator
\begin{equation*}
	\begin{split}
\widehat \psi(a) &= \dfrac{1}{n} \sum\limits_{i=1}^n \Biggl\{ \widehat b_a(X_{1i}, S_i) + \dfrac{I(D = 1)}{c(X_{1i}, S_i)} \Big\{ \widehat g_a(X_i)  - \widehat b_a(S_i, X_{1i}) \Big\} + \dfrac{I(S_i = 1, A_i =a )}{\widehat p(X_i) e_a(X_i)} \Big\{ Y_i  - \widehat g_a(X_i) \Big\}  \Biggl\} \\
	&= \dfrac{1}{n} \sum\limits_{i=1}^n \left\{ \rule{0cm}{0.9cm} \Bigg\{ 1 - \dfrac{I(D_i = 1)}{c(X_{1i}, S_i)} \Bigg\} \widehat b_a(X_{1i}, S_i)  + \dfrac{I(D_i = 1)}{c(X_{1i}, S_i)} \widehat g_a(X_i)  + \dfrac{I(S_i = 1, A_i =a )}{\widehat p(X_i)  e_a(X_i)} \Big\{ Y_i  - \widehat g_a(X_i) \Big\}  \right\},
	\end{split}
\end{equation*}
where $c(X_{1}, S)$ and $ e_a(X)$ are known by design (or, alternatively, can be consistently estimated).

%%%%%%%%%%%%%%%%%%%%%%%%%%%%%%%%%%%%%%
\subsection{Double robustness}
%%%%%%%%%%%%%%%%%%%%%%%%%%%%%%%%%%%%%%

We now consider the behavior of the estimator in Section \ref{subsec:appendix_estiamtor}, using the ``true'' sampling probability, $c(X_1, S) = \Pr[D=1|X_1, S]$, and the ``true'' probability of treatment among randomized individual, $e_a(X) = \Pr[A = a | X, S = 1]$.

Suppose that $\widehat b_a(X_1, S) $, $\widehat g_a(X) $, and $\widehat p(X) $ have well-defined limiting values, which we denote as $b_a^*(X_1, S) $, $ g_a^*(X) $, and $ p^*(X) $, respectively. %Furthermore, if the sampling probability and the probability of treatment among randomized individuals are estimated, suppose that $\widehat c(X_1, S)$ and $\widehat e_a(X) $ also have well-defined limiting values which we denote as $ c^*(X_1, S)$ and $ e_a^*(X) $. 

\begin{proposition}
In the nested trial design with sub-sampling of non-randomized individuals, $\widehat \psi(a)$ is doubly robust in the sense that, $\widehat \psi(a) \overset{p}{\longrightarrow} \psi(a) = \E \big[ \E[Y | X, S = 1, A = a] \big],$ when either $\widehat g_a(X) \overset{p}{\longrightarrow}  g_a^*(X) =  \E[Y|X, S = 1, A = a]$ or $\widehat p(X) \overset{p}{\longrightarrow} p^*(X) = \Pr[S = 1 | X].$
\end{proposition}

\begin{proof}
As $n \rightarrow \infty$, we have that 
\begin{equation}\label{eq:estimator_limiting_value}
\widehat \psi(a) \overset{p}{\longrightarrow} \E  \left[ \rule{0cm}{0.9cm} \Bigg\{ 1 - \dfrac{I(D = 1)}{c(X_{1}, S)} \Bigg\} b_a^*(X_{1}, S)  + \dfrac{I(D = 1)}{c(X_{1}, S)} g_a^*(X)  + \dfrac{I(S = 1, A =a )}{p^*(X) e_a(X)} \Big\{ Y  - g_a^*(X) \Big\}  \right].
\end{equation}

First, we note that by design, $$\E\left[ \rule{0cm}{0.9cm} \Bigg\{ 1 - \dfrac{I(D = 1)}{c(X_{1}, S)} \Bigg\} b_a^*(X_{1}, S) \right] = 0, $$ regardless of whether $b_a^*(X_{1}, S) = \E \big[ \E[Y | X, S = 1, A = a] \big| X_1, S, D = 1 \big]$. 

Next, we study the expectation of the remaining two terms in (\ref{eq:estimator_limiting_value}) by examining cases.

\vspace{0.1in}
\noindent
\emph{Case 1: $g_a^*(X) = \E[Y | X, S = 1, A = a]$, but $p^*(X) \ne \Pr[S = 1 | X]$:} We have that 
\begin{equation*}
	\begin{split}
		\E  \Biggl[ \dfrac{I(D = 1)}{c(X_{1}, S)} g_a^*(X) & + \dfrac{I(S = 1, A =a )}{p^*(X) e_a(X)} \Big\{ Y  - g_a^*(X) \Big\}  \Biggl] \quad\quad\quad\quad\quad\quad\quad\quad\quad\quad\quad\quad\quad\\
		&= \E  \Biggl[  \dfrac{I(D = 1)}{c(X_{1}, S)} \E[Y | X, S = 1 , A = a ] \Biggl] \\
		&= \E  \left[ \rule{0cm}{0.9cm} \dfrac{\E[Y | X, S = 1 , A = a ]}{c(X_{1}, S)}  \E [ I(D = 1) | X , A, Y, S   \Bigg]  \right] \\
		&= \E  \big[  \E[Y | X, S = 1 , A = a ] \big]. 
	\end{split}
\end{equation*}
Thus, if $g_a^*(X) = \E[Y | X, S = 1, A = a]$, then $\widehat \psi(a) \overset{p}{\longrightarrow} \E  \big[  \E[Y | X, S = 1 , A = a ] \big]$.

\vspace{0.1in}
\noindent
\emph{Case 2: $p^*(X) = \Pr[S = 1 | X]$, but $ g_a^*(X) \ne \E[Y | X, S = 1, A = a]$:} We have that 
\begin{equation*}
	\begin{split}
		\E  \Biggl[ \dfrac{I(D = 1)}{c(X_{1}, S)} g_a^*(X) & + \dfrac{I(S = 1, A =a )}{p^*(X) e_a^*(X)} \Big\{ Y  - g_a^*(X) \Big\}  \Biggl] \quad\quad\quad\quad\quad\quad\quad\quad\quad\quad\quad\quad\quad\\
		&= \E  \left[ \rule{0cm}{0.9cm} \Biggl\{     \dfrac{I(D = 1)}{c(X_{1}, S)}  - \dfrac{I(S = 1, A =a )}{p^*(X) e_a(X)}   \Biggl\}  g_a^*(X) + \dfrac{I(S = 1, A =a )}{p^*(X) e_a(X)} Y   \right] \\
		&= \E  \Biggl[ \dfrac{I(S = 1, A =a )}{\Pr[S =1 | X] \Pr[A =a | X, S = 1]} Y   \Biggl] \\
		&= \E  \big[  \E[Y | X, S = 1 , A = a ] \big]. 
	\end{split}
\end{equation*}
Thus, if $p^*(X)  = \Pr[S = 1 | X]$, then $\widehat \psi(a) \overset{p}{\longrightarrow} \E  \big[  \E[Y | X, S = 1 , A = a ] \big]$.

Taken together, \emph{Cases 1 and 2} establish the double robustness of $\psi(a).$

\end{proof}

\begin{remark}
Consistently estimating the sampling probability, $c(X_1, S)$, and the probability of treatment among randomized individuals, $e_a(X)$, does not affect the double robustness of the estimator in Section \ref{subsec:appendix_estiamtor}. The reason is that these probabilities are under the control of the investigators and it is always possible to select estimators $\widehat c(X_1, S)$ and $\widehat e_a(X) $, that have well-defined limiting values, $ c^*(X_1, S)$ and $ e_a^*(X) $, respectively, such that $$\widehat c(X_1, S) \overset{p}{\longrightarrow} c^*(X_1, S) = \Pr[D = 1 | X_1, S]$$ and $$\widehat e_a(X) \overset{p}{\longrightarrow} e_a^*(X) = \Pr[A = a | X, S = 1].$$ 
\end{remark}

%%%%%%%%%%%%%%%%%%%%%%%%%%%%%%%%%%%%%%
\clearpage
\section{Asymptotic distribution}\label{appendix:C_asymptotic_distribution}
%%%%%%%%%%%%%%%%%%%%%%%%%%%%%%%%%%%%%%
\setcounter{equation}{0}
\renewcommand{\theequation}{C.\arabic{equation}}

Recall that 
\[
\widehat \psi(a) = \frac{1}{n} \sum_{i=1}^n \left\{  \widehat b_a(X_{1i},S_i) + \dfrac{I(D_i = 1)}{c(X_{1i}, S_i)}\Big\{\widehat g_a(X_i) - \widehat b_a(X_{1i},S_i)\Big\} + \dfrac{I(S_i =1, A_i =a)}{\widehat p(X_i) e_a(X_i)} \Big\{Y_i - \widehat g_a(X_i)\Big\} \right\}. 
\]
As before, $b_a^*(X_1,S), g_a^*(X)$, and $p^*(X)$ denote the asymptotic limits of the potentially misspecified models $\widehat b_a(X_1,S)$, $\widehat g_a(X)$, and $\widehat p(X)$, respectively. For any functions $b_a'(X_1,S)$, $g_a'(X)$, and $p'(X)$ define
\[
H(b_a',g_a',p') =  b_a'(X_{1},S) + \frac{I(D = 1)}{c(X_{1}, S)}\Big\{g_a'(X) - b_a'(X_{1},S)\Big\} + \frac{I(S =1, A =a)}{p'(X) e_a(X)} \Big\{Y - g_a'(X)\Big\}.
\]
Using standard empirical processes notation \cite{van1996weak}, denote 
\begin{align*}
&\mathbb{P}_n(H(b_a',g_a',p')) \\ &= \dfrac{1}{n} \sum_{i=1}^n \left\{  b_a'(X_{1i},S_i) + \dfrac{I(D_i = 1)}{c(X_{1i}, S_i)}\Big\{g_a'(X_i) - b_a'(X_{1i},S_i)\Big\} + \dfrac{I(S_i =1, A_i =a)}{p'(X_i) e_a(X_i)} \Big\{Y_i - g_a'(X_i)\Big\} \right\}
\end{align*}
and let $\mathbb{G}_n(H(b_a',g_a',p')) = \sqrt{n} (\mathbb{P}_n(H(b_a',g_a',p')) - \E[H(b_a',g_a',p')])$. Note that $\widehat \psi(a) = \mathbb{P}_n(H(\widehat b_a, \widehat g_a, \widehat p))$. 

The derivation of the asymptotic distribution relies on the following additional assumptions:
\begin{enumerate}
\item[A.1] The sequence $H(\widehat b_a,\widehat g_a, \widehat p)$ and the limit $H(b_a^*,g_a^*, p^*)$ fall in a Donsker class \cite{van1996weak}.
\item[A.2] We have $||H(\widehat b_a, \widehat g_a, \widehat p) - H(b_a^*,g_a^*, p^*)||_2 \rightarrow 0$.
\item[A.3] We have $\E[H(b_a^*,g_a^*, p^*)^2] < \infty$.
\end{enumerate}
If $\widehat b_a$, $\widehat g_a$, $\widehat p$, $b_a^*$, $g_a^*$, and $p^*$ all fall in a Donsker class; $p^*$ is uniformly bounded away from zero; and $g_a^*$ is uniformly bounded, then assumption A.1 follows from Corollary $9.31$ in \cite{kosorok2008introduction}. 

The following Proposition gives the asymptotic distribution of the estimator $\widehat \psi(a)$. 
\begin{proposition}
\label{Prop:Inf}
Under the assumptions made, in the nested trial design with sub-sampling of non-randomized individuals, 
\begin{equation*}
\sqrt{n}(\widehat \psi(a) - \psi(a)) =  \mathbb{G}_n(H(b_a^*, g_a^*, p^*)) + R + o_P(1),
\end{equation*}
where $\mathbb{G}_n(H(b_a^*, g_a^*, p^*))$ is asymptotically normal and
\[
R \leq \sqrt{n}O_P\left(||\widehat g_a(X) - \E[Y|X,S=1, A=a]||_2 ||\widehat \Pr[S=1,A=a|X] - \Pr[S=1,A=a|X]||_2\right).
\]
\end{proposition}
\begin{proof}
Decompose $\sqrt{n}(\widehat \psi(a) - \psi(a))$ as
\begin{align*} 
\sqrt{n}(\widehat \psi(a) - \psi(a)) &= \left(\mathbb{G}_n(H(\widehat b_a, \widehat g_a, \widehat p)) - \mathbb{G}_n(H(b_a^*,g_a^*, p^*)) \right)  + \mathbb{G}_n(H(b_a^*, g_a^*, p^*)) \\
&\quad\quad\quad+ \sqrt{n}\left(\E[H(\widehat b_a,\widehat g_a, \widehat p)]) - \psi(a)\right).
\end{align*}
All convergence results presented here are in terms of $n\rightarrow \infty$. The proof relies on working with each of the terms on the right-hand-side of the above equation separately. 

For the first term in the decomposition of $\sqrt{n}(\widehat \psi(a) - \psi(a))$, by the Donsker property \cite{van1996weak,  pollard2012convergence} of $H(\widehat b_a, \widehat g_a, \widehat p)$ and $H(b_a^*,g_a^*, p^*)$, we have 
\[
\left(\mathbb{G}_n(H(\widehat b_a,\widehat g_a, \widehat p)) - \mathbb{G}_n(H(b_a^*,g_a^*, p^*)) \right) = o_P(1).
\]
The second term, $\mathbb{G}_n(H(b_a^*, g_a^*, p^*))$, in the decomposition is asymptotically normal by the central limit theorem. Hence, the asymptotic distribution of $\widehat \psi(a)$ depends on the behavior of the third term, $\sqrt{n}\left(\E[H(\widehat b_a, \widehat g_a, \widehat p)] - \psi(a)\right)$. We re-write the third term as
\begin{align*}
&\sqrt{n}\left(\E[H(\widehat b_a, \widehat g_a, \widehat p)]) - \psi(a) \right) \\ 
&\quad =  \sqrt{n} \left( \E\left[\widehat b_a(X_{1},S) + \frac{I(D = 1)}{c(X_{1}, S)}\Big\{\widehat g_a(X) - \widehat b_a(X_{1},S)\Big\} + \frac{I(S =1, A =a)}{\widehat p(X) e_a(X)} \Big\{Y - \widehat g_a(X)\Big\} \right] - \psi(a) \right)\\ 
&\quad = \underbrace{\sqrt{n} \E\left[\widehat b_a(X_{1},S) \left\{1 - \frac{I(D = 1)}{c(X_{1}, S)}\right\}\right]}_{R_1} \\
&\quad\quad\quad + \underbrace{\sqrt{n} \E\left[\frac{I(D = 1)}{c(X_{1}, S)} \widehat g_a(X) + \frac{I(S =1, A =a)}{\widehat p(X) e_a(X)} \Big\{Y - \widehat g_a(X)\Big\} - \psi(a) \right]}_{R}.
\end{align*}
First, we rewrite $R_1$ as
\begin{align*}
R_1 &= \sqrt{n} \E\left[b_a^*(X_{1},S) \left\{1 - \dfrac{I(D = 1)}{c(X_{1}, S)}\right\}\right] + \sqrt{n} \E\left[\Big\{\widehat b_a(X_{1},S) - b_a^*(X_{1},S)\Big\} \left\{1 - \frac{I(D = 1)}{c(X_{1}, S)}\right\}\right].
\end{align*}
As $\E\left[\dfrac{I(D = 1)}{c(X_{1}, S)}\right] = 1$, the first term on the right hand side is equal to zero. Assuming that $\sqrt{n}(\widehat b_a(X_{1},S) - b_a^*(X_{1},S)) = O_P(1)$, the second term is also $o_P(1)$. These arguments do not assume that the model $\widehat b_a(X_1,S)$ is correctly specified and the required $\sqrt{n}$ convergence (to a potentially misspecified limit) can always be obtained using a parametric model for $\widehat b_a(X_1, S)$.

Next, we rewrite $R$ as
\begin{align*}
R &= \sqrt{n} \E\left[\frac{I(D = 1)}{c(X_{1}, S)} \widehat g_a(X) + \frac{I(S =1, A =a)}{\widehat p(X) e_a(X)} \Big\{Y - \widehat g_a(X)\Big\} - \psi(a) \right] \\
&= \sqrt{n} \E\left[\Big\{\widehat g_a(X) - \E[Y|X,S=1, A=a]\Big\} + \frac{\Pr[S=1,A=a|X]}{\widehat \Pr[S=1,A=a|X]} \Big\{\E[Y|X,S=1, A=a] - \widehat g_a(X)\Big\} \right] \\
&= \sqrt{n} \E\left[\Big\{\widehat g_a(X) - \E[Y|X,S=1, A=a]\Big\} \left\{1 - \frac{\Pr[S=1,A=a|X]}{\widehat \Pr[S=1,A=a|X]}\right\} \right] \\
& \leq \sqrt{n}O_P(||\widehat g_a(X) - \E[Y|X,S=1, A=a]||_2 ||\widehat \Pr[S=1,A=a|X] - \Pr[S=1,A=a|X]||_2),
\end{align*} 
where the last line follows from the Cauchy-Schwarz inequality and the boundedness of $\widehat \Pr[S=1,A=a|X]$ away from zero.
\end{proof}

\begin{remark}
The term $R$ in Proposition \ref{Prop:Inf} identifies how the estimators of the nuisance parameters $\widehat g_a(X)$ and $\widehat{Pr}[S=1,A=a|X]$ affect the distribution of $\widehat \psi(a)$. If the nuisance parameters converge to the true population parameters at a rate
\begin{equation*}
\sqrt{n}||\widehat g_a(X) - \E[Y|X,S=1, A=a]||_2 ||\Pr[S=1,A=a|X] - \widehat \Pr[S=1, A=a|X]||_2 = o_P(1),
\end{equation*}
the term $R$ in the proposition does not contribute to the asymptotic variance of the estimator. If the estimators $\widehat g_a(X)$ and $\widehat p(X)$ come from the class of generalized linear models, the estimators are Donsker and have a fast enough convergence rate for $R$ to be $o_p(1)$ if both models are correct and $O_P(1)$ if at least one model is correct (the doubly robustness property previously discussed). 
If more data adaptive estimators that do not necessarily converge at a fast enough rate are used to calculate the nuisance parameters $\widehat g_a$, and $\widehat p$, sample splitting can be used to control the behavior of $R$ \cite{chernozhukov2018double}.
\end{remark}

%%%%%%%%%%%%%%%%%%%%%%%%%%%%%%%%%%%%%%
\clearpage
\section{Asymptotic efficiency}\label{appendix:D_asymptotic_efficiency}
%%%%%%%%%%%%%%%%%%%%%%%%%%%%%%%%%%%%%%
\setcounter{equation}{0}
\renewcommand{\theequation}{D.\arabic{equation}}

As we have shown previously \cite{dahabreh2018generalizing}, the estimator in (\ref{eq:estimator_nosub}), in the absence of sub-sampling, when the models for $g_{a}(X)$ and $p(X)$ are correctly specified, has asymptotic variance
\begin{equation}\label{eq:avar1}
	\mbox{AVar}_1 = n^{-1} \Biggl\{  \E \left[ \dfrac{v_{a}(X)}{p(X) e_{a}(X)}  \right] + \mbox{Var} \big[ g_{a}(X) \big] \Biggl\},
\end{equation}
where $v_{a}(X) = \mbox{Var}[Y|X, S = 1, A = a]$; $p(X)$, $e_{a}(X)$, and $g_{a}(X)$ are as defined above; and all quantities are evaluated at the true law.

Furthermore, using the influence function result in Appendix \ref{appendix:B_estimation_robustness}, we obtain, via routine algebraic manipulation,
\begin{equation*}
	\begin{split}
	\E \big [\{ \psi_0^1(a) \}^2 \big] &= \E \left[ \dfrac{v_{a}(X)}{p(X) e_{a}(X)}  \right] + \mbox{Var} \big[ g_{a}(X) \big] + \E\left[ \dfrac{1 - c(X_1, S)}{c(X_1, S)} \Big\{ g_{a}(X)   -  b_{a}(X_1, S)      \Big\}^2 \right].
	\end{split}
\end{equation*}
Thus, when $g_a(X)$ and $p(X)$ are consistently estimated using correctly specified models (and at sufficiently fast rate), and $b_a(X_1, S)$ is estimated at $\sqrt{n}$-rate (even with a misspecified model), the estimator in (\ref{eq:estimator}) has asymptotic variance
\begin{equation}\label{eq:avar2}
	\begin{split}
	\mbox{AVar}_2 &= \mbox{AVar}_1 + n^{-1}  \E\left[ \dfrac{1 - c(X_1, S)}{c(X_1, S)} \Big\{ g_{a}(X)   -  b_{a}(X_1, S)      \Big\}^2 \right].
	\end{split}
\end{equation}
Comparing (\ref{eq:avar1}) and (\ref{eq:avar2}), we see that $$\mbox{AVar2} \geq \mbox{AVar1}.$$

%%%%%%%%%%%%%%%%%%%%%%%%%%%%%%%%%%%%%%
\clearpage
\section{Additional information about the simulation study}\label{appendix:E_simulation_details}
%%%%%%%%%%%%%%%%%%%%%%%%%%%%%%%%%%%%%%

\setcounter{table}{0}
\renewcommand{\thetable}{E.\arabic{table}}

\setcounter{figure}{0}
\renewcommand{\thefigure}{E.\arabic{figure}}

{
\singlespacing

\begin{table}[H]
\centering
\caption{Scenarios considered in the simulation study, for covariate dependent sampling probabilities and continuous $X_1$. For each cohort sample size ($n$) and average trial sample size, we provide the $\gamma_0$ values that result in the desired marginal probability of trial participation, $\Pr[S=1]$, and the $\zeta_0$ values that result in the desired marginal sampling probabilities among non-randomized individuals, $\Pr[D = 1 | X_1, S = 0]$.} \label{table:sample_size_scenarios_cont}
\begin{tabular}{llccc}
\toprule
\multirow{2}{*}{\begin{tabular}[c]{@{}l@{}}Average\\ trial size\end{tabular}} & \multicolumn{1}{c}{\multirow{2}{*}{$n$}} & \multirow{2}{*}{$\Pr[S=1]$} & \multirow{2}{*}{$\gamma_0$} & \multirow{2}{*}{\begin{tabular}[c]{@{}c@{}}$\zeta_0$ values for marginal sampling probabilities\\ ranging from 0.1 to 0.9, in steps of 0.1\end{tabular}}    \\
                                                                              & \multicolumn{1}{c}{}                     &                             &                             &                                                                                                                                                             \\ \midrule
\multirow{3}{*}{1000}                                                         & 2000                                      & 0.5                         & 0                           & \begin{tabular}[c]{@{}c@{}}-2.1953125, -1.2929688, -0.6761070,\\ -0.1483765, 0.3237305, 0.8099365,\\ 1.3345490, 1.9550781, 2.8554688\end{tabular}           \\ \cline{2-5} 
                                                                              & 5000                                      & 0.2                         & -2.055969                   & \begin{tabular}[c]{@{}c@{}}-2.3974609, -1.4904175, -0.8675537,\\ -0.3417969, 0.1408870, 0.6245117,\\ 1.1464232, 1.7731247, 2.6875000\end{tabular}           \\ \cline{2-5} 
                                                                              & 10,000                                    & 0.1                         & -3.154297                   & \begin{tabular}[c]{@{}c@{}}-2.47167969, -1.56103516, -0.93359375,\\  -0.40990990, 0.07421875, 0.56357574,\\ 1.08593750, 1.71679688, 2.62744141\end{tabular} \\ \hline
\multirow{3}{*}{2000}                                                         & 5000                                      & 0.4                         & -0.612793                   & \begin{tabular}[c]{@{}c@{}}-2.2607422, -1.3611903, -0.7357330,\\  -0.2153320, 0.2639160, 0.7441406,\\ 1.2669601, 1.8906250, 2.7956066\end{tabular}          \\ \cline{2-5} 
                                                                              & 10,000                                    & 0.2                         & -2.055969                   & \begin{tabular}[c]{@{}c@{}}-2.3974609, -1.4904175, -0.8675537,\\  -0.3417969, 0.1408870, 0.6245117, \\ 1.1464232, 1.7731247, 2.6875000\end{tabular}         \\ \cline{2-5} 
                                                                              & 20,000                                    & 0.1                         & -3.154297                   & \begin{tabular}[c]{@{}c@{}}-2.47167969, -1.56103516, -0.93359375,\\ -0.40990990, 0.07421875, 0.56357574,\\ 1.08593750, 1.71679688, 2.62744141\end{tabular}  \\ \bottomrule
\end{tabular}
\end{table}

}

{
\singlespacing

\begin{table}[H]
\centering
\caption{Scenarios considered in the simulation study, for covariate dependent sampling probabilities and binary $X_1$. For each cohort sample size ($n$) and average trial sample size, we provide the $\gamma_0$ values that result in the desired marginal probability of trial participation, $\Pr[S=1]$, and the $\zeta_0$ values that result in the desired marginal sampling probabilities among non-randomized individuals, $\Pr[D = 1 | X_1, S = 0]$.} \label{table:sample_size_scenarios_binary}
\begin{tabular}{llccc}
\toprule
\multirow{2}{*}{\begin{tabular}[c]{@{}l@{}}Average\\ trial size\end{tabular}} & \multicolumn{1}{c}{\multirow{2}{*}{$n$}} & \multirow{2}{*}{$\Pr[S=1]$} & \multirow{2}{*}{$\gamma_0$} & \multicolumn{1}{c}{\multirow{2}{*}{\begin{tabular}[c]{@{}c@{}}$\zeta_0$ values for marginal sampling probabilities\\ ranging from 0.1 to 0.9, in steps of 0.1\end{tabular}}} \\
                                                                              & \multicolumn{1}{c}{}                     &                             &                             & \multicolumn{1}{c}{}                                                                                                                                                  \\ \midrule
\multirow{3}{*}{1000}                                                         & 2000                                      & 0.5                         & -0.4973936                  & \begin{tabular}[c]{@{}l@{}}-2.70117188, -1.87109375, -1.30666184, \\ -0.83532715, -0.40234375, 0.02183144, \\ 0.48690367, 1.04687500, 1.88330555\end{tabular}          \\ \cline{2-5} 
                                                                              & 5000                                      & 0.2                         & -2.4145508                  & \begin{tabular}[c]{@{}l@{}}-2.7578125, -1.9258423, -1.3582602,\\ -0.8906250, -0.4609375, -0.0312500,\\ 0.4363470, 0.9983544, 1.8328857\end{tabular}                    \\ \cline{2-5} 
                                                                              & 10,000                                    & 0.1                         & -3.460083                   & \begin{tabular}[c]{@{}l@{}}-2.77343750, -1.93980408, -1.37890625, \\ -0.91027832, -0.48046875, -0.05080032,\\ 0.41790675, 0.98059082, 1.81640625\end{tabular}          \\ \hline
\multirow{3}{*}{2000}                                                         & 5000                                      & 0.4                         & -1.072715                   & \begin{tabular}[c]{@{}l@{}}-2.7205811, -1.8925781, -1.3217773,\\ -0.8562012, -0.4282227, 0.0078125,\\ 0.4677734, 1.0312500, 1.8676951\end{tabular}                     \\ \cline{2-5} 
                                                                              & 10,000                                    & 0.2                         & -2.4145508                  & \begin{tabular}[c]{@{}l@{}}-2.7578125, -1.9258423, -1.3582602, \\ -0.8906250, -0.4609375 -0.0312500,\\ 0.4363470, 0.9983544, 1.8328857\end{tabular}                    \\ \cline{2-5} 
                                                                              & 20,000                                    & 0.1                         & -3.460083                   & \begin{tabular}[c]{@{}l@{}}-2.77343750, -1.93980408, -1.37890625,\\ -0.91027832, -0.48046875, -0.05080032,\\ 0.41790675, 0.98059082, 1.81640625\end{tabular}           \\ \bottomrule
\end{tabular}
\end{table}

}

\vspace{0.2in}
\noindent
\emph{Note:} simple random sampling of non-randomized individuals, numerical methods are not needed to solve for $\zeta_0$ and the numerical solutions for $\gamma_0$ are the same as in the above two tables.

% BIAS, CONT, COVDEP

\clearpage

\begin{sidewaystable}[]
\caption{Bias estimates for simulation scenarios with continuous $Z_1$ and covariate dependent sampling.} \label{table:BIAS_continuous_results}
\resizebox{\textwidth}{!}{
\begin{tabular}{ccccccccccccc}
\toprule
\multirow{2}{*}{\begin{tabular}[c]{@{}l@{}}Target\\ parameter\end{tabular}} & \multirow{2}{*}{\begin{tabular}[c]{@{}l@{}}Average\\ trial size\end{tabular}} & \multirow{2}{*}{$n$} & \multicolumn{10}{c}{Marginal sampling probability among non-randomized individuals}                                                                                                                   \\ \cline{4-13}
                                                                              &               &       & 0.1            & 0.2            & 0.3            & 0.4            & 0.5            & 0.6            & 0.7            & 0.7            & 0.9            & 1                                            \\ \midrule
\multirow{6}{*}{$\psi(1)$} &	1000 &	2000 &	0.0001 &	-0.0002 &	-0.0002 &	-0.0002 &	-0.0004 &	-0.0005 &	-0.0003 &	-0.0004 &	-0.0003 &	-0.0003 \\
 & 	1000 &	5000 &	0.0014 &	0.0023 &	0.0015 &	0.0013 &	0.0009 &	0.0010 &	0.0011 &	0.0010 &	0.0011 &	0.0011 \\
 & 	1000 &	10000 &	-0.0035 &	-0.0020 &	-0.0021 &	-0.0022 &	-0.0025 &	-0.0022 &	-0.0023 &	-0.0023 &	-0.0023 &	-0.0023 \\
 & 	2000 &	5000 &	-0.0001 &	0.0005 &	0.0002 &	-0.0002 &	-0.0001 &	0.0003 &	0.0001 &	0.0002 &	0.0001 &	0.0001 \\
 & 	2000 &	10000 &	0.0007 &	0.0001 &	0.0002 &	0.0001 &	0.0002 &	0.0002 &	0.0002 &	0.0003 &	0.0002 &	0.0002 \\
 & 	2000 &	20000 &	0.0015 &	0.0020 &	0.0016 &	0.0017 &	0.0016 &	0.0018 &	0.0018 &	0.0017 &	0.0017 &	0.0017 \\ \midrule
\multirow{6}{*}{$\psi(0)$} &	1000 &	2000 &	0.0011 &	0.0006 &	-0.0000 &	-0.0004 &	-0.0005 &	-0.0007 &	-0.0004 &	-0.0006 &	-0.0004 &	-0.0004 \\
 & 	1000 &	5000 &	0.0019 &	0.0022 &	0.0016 &	0.0015 &	0.0011 &	0.0011 &	0.0011 &	0.0011 &	0.0011 &	0.0012 \\
 & 	1000 &	10000 &	-0.0030 &	-0.0010 &	-0.0016 &	-0.0015 &	-0.0017 &	-0.0014 &	-0.0015 &	-0.0015 &	-0.0014 &	-0.0014 \\
 & 	2000 &	5000 &	-0.0004 &	0.0007 &	-0.0002 &	-0.0003 &	-0.0003 &	0.0000 &	-0.0001 &	-0.0001 &	-0.0002 &	-0.0002 \\
 & 	2000 &	10000 &	0.0019 &	0.0010 &	0.0010 &	0.0009 &	0.0008 &	0.0009 &	0.0010 &	0.0010 &	0.0011 &	0.0009 \\
 & 	2000 &	20000 &	-0.0018 &	-0.0004 &	-0.0011 &	-0.0010 &	-0.0011 &	-0.0011 &	-0.0010 &	-0.0010 &	-0.0011 &	-0.0011 \\  \midrule
\multirow{6}{*}{$\psi(1) - \psi(0)$} &	1000 &	2000 &	-0.0010 &	-0.0008 &	-0.0002 &	0.0002 &	0.0001 &	0.0002 &	0.0001 &	0.0002 &	0.0001 &	0.0001 \\
 & 	1000 &	5000 &	-0.0005 &	0.0002 &	-0.0001 &	-0.0002 &	-0.0003 &	-0.0000 &	-0.0000 &	-0.0001 &	0.0000 &	-0.0001 \\
 & 	1000 &	10000 &	-0.0005 &	-0.0010 &	-0.0005 &	-0.0008 &	-0.0009 &	-0.0008 &	-0.0008 &	-0.0008 &	-0.0009 &	-0.0009 \\
 & 	2000 &	5000 &	0.0003 &	-0.0002 &	0.0004 &	0.0001 &	0.0002 &	0.0003 &	0.0003 &	0.0003 &	0.0003 &	0.0003 \\
 & 	2000 &	10000 &	-0.0012 &	-0.0009 &	-0.0009 &	-0.0008 &	-0.0006 &	-0.0007 &	-0.0008 &	-0.0007 &	-0.0009 &	-0.0007 \\
 & 	2000 &	20000 &	0.0034 &	0.0024 &	0.0027 &	0.0027 &	0.0028 &	0.0028 &	0.0027 &	0.0027 &	0.0028 &	0.0028 \\ \bottomrule
\end{tabular}
}
\end{sidewaystable}

% BIAS, CONT, SRS

\clearpage

\begin{sidewaystable}[]
\caption{Bias estimates for simulation scenarios with continuous $Z_1$ and simple random sampling.} \label{table:BIAS_continuousSRS_results}
\resizebox{\textwidth}{!}{
\begin{tabular}{ccccccccccccc}
\toprule
\multirow{2}{*}{\begin{tabular}[c]{@{}l@{}}Target\\ parameter\end{tabular}} & \multirow{2}{*}{\begin{tabular}[c]{@{}l@{}}Average\\ trial size\end{tabular}} & \multirow{2}{*}{$n$} & \multicolumn{10}{c}{Marginal sampling probability among non-randomized individuals}                                                                                                                   \\ \cline{4-13}
                                                                              &               &       & 0.1            & 0.2            & 0.3            & 0.4            & 0.5            & 0.6            & 0.7            & 0.7            & 0.9            & 1                                            \\ \midrule
\multirow{6}{*}{$\psi(1)$} &	1000 &	2000 &	0.0002 &	-0.0007 &	-0.0003 &	-0.0003 &	-0.0003 &	-0.0004 &	-0.0003 &	-0.0004 &	-0.0003 &	-0.0003 \\
 & 	1000 &	5000 &	0.0016 &	0.0022 &	0.0015 &	0.0011 &	0.0008 &	0.0011 &	0.0011 &	0.0011 &	0.0011 &	0.0011 \\
 & 	1000 &	10000 &	-0.0020 &	-0.0023 &	-0.0021 &	-0.0021 &	-0.0023 &	-0.0023 &	-0.0024 &	-0.0023 &	-0.0023 &	-0.0023 \\
 & 	2000 &	5000 &	-0.0009 &	-0.0017 &	-0.0012 &	-0.0012 &	-0.0012 &	-0.0012 &	-0.0013 &	-0.0012 &	-0.0012 &	-0.0012 \\
 & 	2000 &	10000 &	-0.0003 &	-0.0001 &	0.0000 &	0.0000 &	-0.0000 &	-0.0002 &	-0.0001 &	-0.0001 &	-0.0001 &	-0.0001 \\
 & 	2000 &	20000 &	0.0032 &	0.0033 &	0.0035 &	0.0035 &	0.0035 &	0.0032 &	0.0034 &	0.0034 &	0.0034 &	0.0034 \\ \midrule
\multirow{6}{*}{$\psi(0)$} &	1000 &	2000 &	0.0012 &	-0.0012 &	-0.0007 &	-0.0003 &	-0.0003 &	-0.0005 &	-0.0002 &	-0.0005 &	-0.0004 &	-0.0004 \\
 & 	1000 &	5000 &	0.0014 &	0.0016 &	0.0018 &	0.0015 &	0.0010 &	0.0011 &	0.0011 &	0.0012 &	0.0011 &	0.0012 \\
 & 	1000 &	10000 &	-0.0024 &	-0.0013 &	-0.0019 &	-0.0015 &	-0.0016 &	-0.0014 &	-0.0014 &	-0.0015 &	-0.0014 &	-0.0014 \\
 & 	2000 &	5000 &	0.0018 &	0.0003 &	0.0013 &	0.0008 &	0.0010 &	0.0009 &	0.0009 &	0.0009 &	0.0010 &	0.0009 \\
 & 	2000 &	10000 &	0.0008 &	0.0007 &	0.0008 &	0.0006 &	0.0005 &	0.0004 &	0.0006 &	0.0006 &	0.0006 &	0.0006 \\
 & 	2000 &	20000 &	0.0038 &	0.0041 &	0.0038 &	0.0041 &	0.0036 &	0.0036 &	0.0038 &	0.0037 &	0.0038 &	0.0037 \\ \midrule
\multirow{6}{*}{$\psi(1) - \psi(0)$} &	1000 &	2000 &	-0.0009 &	0.0005 &	0.0004 &	-0.0000 &	-0.0000 &	0.0002 &	-0.0001 &	0.0001 &	0.0001 &	0.0001 \\
 & 	1000 &	5000 &	0.0002 &	0.0006 &	-0.0003 &	-0.0004 &	-0.0002 &	0.0000 &	0.0001 &	-0.0001 &	-0.0000 &	-0.0001 \\
 & 	1000 &	10000 &	0.0004 &	-0.0010 &	-0.0002 &	-0.0006 &	-0.0007 &	-0.0009 &	-0.0009 &	-0.0009 &	-0.0009 &	-0.0009 \\
 & 	2000 &	5000 &	-0.0028 &	-0.0020 &	-0.0025 &	-0.0021 &	-0.0022 &	-0.0021 &	-0.0023 &	-0.0021 &	-0.0022 &	-0.0022 \\
 & 	2000 &	10000 &	-0.0011 &	-0.0008 &	-0.0007 &	-0.0005 &	-0.0006 &	-0.0005 &	-0.0007 &	-0.0007 &	-0.0008 &	-0.0007 \\
 & 	2000 &	20000 &	-0.0006 &	-0.0008 &	-0.0003 &	-0.0006 &	-0.0001 &	-0.0004 &	-0.0004 &	-0.0003 &	-0.0004 &	-0.0003 \\ \bottomrule
\end{tabular}
}
\end{sidewaystable}

% BIAS, BINARY, COVDEP

\clearpage

\begin{sidewaystable}[]
\caption{Bias estimates for simulation scenarios with binary $Z_1$ and covariate-dependent sampling.} \label{table:BIAS_binary_results}
\resizebox{\textwidth}{!}{
\begin{tabular}{ccccccccccccc}
\toprule
\multirow{2}{*}{\begin{tabular}[c]{@{}l@{}}Target\\ parameter\end{tabular}} & \multirow{2}{*}{\begin{tabular}[c]{@{}l@{}}Average\\ trial size\end{tabular}} & \multirow{2}{*}{$n$} & \multicolumn{10}{c}{Marginal sampling probability among non-randomized individuals}                                                                                                                   \\ \cline{4-13}
                                                                              &               &       & 0.1            & 0.2            & 0.3            & 0.4            & 0.5            & 0.6            & 0.7            & 0.7            & 0.9            & 1                                            \\ \midrule
\multirow{6}{*}{$\psi(1)$} &	1000 &	2000 &	0.0007 &	0.0011 &	0.0005 &	0.0009 &	0.0010 &	0.0007 &	0.0009 &	0.0008 &	0.0008 &	0.0007 \\
 & 	1000 &	5000 &	0.0001 &	-0.0000 &	-0.0002 &	-0.0005 &	-0.0003 &	-0.0003 &	-0.0003 &	-0.0002 &	-0.0002 &	-0.0002 \\
 & 	1000 &	10000 &	-0.0011 &	-0.0014 &	-0.0011 &	-0.0013 &	-0.0012 &	-0.0012 &	-0.0011 &	-0.0013 &	-0.0012 &	-0.0012 \\
 & 	2000 &	5000 &	-0.0004 &	-0.0000 &	-0.0004 &	0.0000 &	-0.0000 &	-0.0001 &	-0.0002 &	-0.0002 &	-0.0001 &	-0.0001 \\
 & 	2000 &	10000 &	0.0018 &	0.0017 &	0.0018 &	0.0015 &	0.0016 &	0.0017 &	0.0017 &	0.0015 &	0.0016 &	0.0016 \\
 & 	2000 &	20000 &	0.0008 &	0.0009 &	0.0008 &	0.0008 &	0.0006 &	0.0006 &	0.0007 &	0.0006 &	0.0006 &	0.0007 \\ \midrule
\multirow{6}{*}{$\psi(0)$} &	1000 &	2000 &	0.0002 &	0.0016 &	0.0006 &	0.0014 &	0.0012 &	0.0010 &	0.0009 &	0.0008 &	0.0009 &	0.0009 \\
 & 	1000 &	5000 &	-0.0000 &	-0.0008 &	-0.0004 &	-0.0008 &	-0.0007 &	-0.0006 &	-0.0007 &	-0.0006 &	-0.0006 &	-0.0006 \\
 & 	1000 &	10000 &	0.0003 &	-0.0007 &	0.0001 &	-0.0002 &	-0.0001 &	-0.0002 &	-0.0002 &	-0.0002 &	-0.0002 &	-0.0002 \\
 & 	2000 &	5000 &	0.0001 &	0.0008 &	-0.0000 &	0.0004 &	0.0003 &	0.0002 &	0.0003 &	0.0002 &	0.0003 &	0.0003 \\
 & 	2000 &	10000 &	-0.0005 &	-0.0000 &	0.0000 &	-0.0006 &	-0.0004 &	-0.0004 &	-0.0004 &	-0.0006 &	-0.0005 &	-0.0005 \\
 & 	2000 &	20000 &	-0.0002 &	-0.0003 &	-0.0005 &	-0.0004 &	-0.0006 &	-0.0004 &	-0.0005 &	-0.0005 &	-0.0006 &	-0.0005 \\ \midrule
\multirow{6}{*}{$\psi(1) - \psi(0)$} &	1000 &	2000 &	0.0005 &	-0.0005 &	-0.0001 &	-0.0004 &	-0.0002 &	-0.0003 &	0.0000 &	0.0000 &	-0.0001 &	-0.0001 \\
 & 	1000 &	5000 &	0.0001 &	0.0007 &	0.0002 &	0.0002 &	0.0003 &	0.0003 &	0.0004 &	0.0004 &	0.0004 &	0.0003 \\
 & 	1000 &	10000 &	-0.0014 &	-0.0007 &	-0.0012 &	-0.0011 &	-0.0012 &	-0.0010 &	-0.0009 &	-0.0011 &	-0.0010 &	-0.0010 \\
 & 	2000 &	5000 &	-0.0006 &	-0.0008 &	-0.0003 &	-0.0004 &	-0.0003 &	-0.0003 &	-0.0004 &	-0.0004 &	-0.0003 &	-0.0004 \\
 & 	2000 &	10000 &	0.0023 &	0.0018 &	0.0018 &	0.0021 &	0.0020 &	0.0021 &	0.0020 &	0.0021 &	0.0021 &	0.0021 \\
 & 	2000 &	20000 &	0.0010 &	0.0012 &	0.0014 &	0.0012 &	0.0012 &	0.0011 &	0.0012 &	0.0011 &	0.0012 &	0.0012 \\ \bottomrule
\end{tabular}
}
\end{sidewaystable}

% BIAS, BINARY, SRS

\clearpage

\begin{sidewaystable}[]
\caption{Bias estimates for simulation scenarios with binary $Z_1$ and simple random sampling.} \label{table:BIAS_binarySRS_results}
\resizebox{\textwidth}{!}{
\begin{tabular}{ccccccccccccc}
\toprule
\multirow{2}{*}{\begin{tabular}[c]{@{}l@{}}Target\\ parameter\end{tabular}} & \multirow{2}{*}{\begin{tabular}[c]{@{}l@{}}Average\\ trial size\end{tabular}} & \multirow{2}{*}{$n$} & \multicolumn{10}{c}{Marginal sampling probability among non-randomized individuals}                                                                                                                   \\ \cline{4-13}
                                                                              &               &       & 0.1            & 0.2            & 0.3            & 0.4            & 0.5            & 0.6            & 0.7            & 0.7            & 0.9            & 1                                            \\ \midrule
\multirow{6}{*}{$\psi(1)$} &	1000 &	2000 &	0.0005 &	0.0008 &	0.0007 &	0.0010 &	0.0010 &	0.0007 &	0.0009 &	0.0008 &	0.0008 &	0.0007 \\
 & 	1000 &	5000 &	-0.0001 &	-0.0000 &	0.0001 &	-0.0005 &	-0.0004 &	-0.0003 &	-0.0003 &	-0.0002 &	-0.0002 &	-0.0002 \\
 & 	1000 &	10000 &	-0.0011 &	-0.0013 &	-0.0012 &	-0.0013 &	-0.0012 &	-0.0011 &	-0.0012 &	-0.0013 &	-0.0012 &	-0.0012 \\
 & 	2000 &	5000 &	-0.0001 &	-0.0000 &	-0.0004 &	-0.0004 &	-0.0003 &	-0.0007 &	-0.0007 &	-0.0005 &	-0.0006 &	-0.0006 \\
 & 	2000 &	10000 &	0.0028 &	0.0029 &	0.0028 &	0.0028 &	0.0027 &	0.0027 &	0.0027 &	0.0028 &	0.0027 &	0.0028 \\
 & 	2000 &	20000 &	0.0015 &	0.0010 &	0.0013 &	0.0010 &	0.0010 &	0.0010 &	0.0011 &	0.0011 &	0.0011 &	0.0011 \\ \midrule
\multirow{6}{*}{$\psi(0)$} &	1000 &	2000 &	-0.0003 &	0.0011 &	0.0006 &	0.0014 &	0.0011 &	0.0008 &	0.0010 &	0.0008 &	0.0008 &	0.0009 \\
 & 	1000 &	5000 &	0.0003 &	-0.0004 &	-0.0002 &	-0.0007 &	-0.0007 &	-0.0006 &	-0.0006 &	-0.0006 &	-0.0006 &	-0.0006 \\
 & 	1000 &	10000 &	0.0003 &	-0.0005 &	0.0003 &	-0.0003 &	-0.0002 &	-0.0001 &	-0.0001 &	-0.0002 &	-0.0001 &	-0.0002 \\
 & 	2000 &	5000 &	0.0007 &	0.0008 &	0.0006 &	0.0005 &	0.0002 &	0.0001 &	0.0000 &	0.0001 &	0.0002 &	0.0001 \\
 & 	2000 &	10000 &	0.0008 &	0.0014 &	0.0012 &	0.0015 &	0.0013 &	0.0014 &	0.0014 &	0.0013 &	0.0012 &	0.0013 \\
 & 	2000 &	20000 &	0.0028 &	0.0022 &	0.0025 &	0.0023 &	0.0021 &	0.0021 &	0.0022 &	0.0022 &	0.0022 &	0.0022 \\ \midrule
\multirow{6}{*}{$\psi(1) - \psi(0)$} &	1000 &	2000 &	0.0007 &	-0.0003 &	0.0001 &	-0.0004 &	-0.0001 &	-0.0002 &	-0.0000 &	-0.0000 &	-0.0000 &	-0.0001 \\
 & 	1000 &	5000 &	-0.0004 &	0.0004 &	0.0003 &	0.0001 &	0.0004 &	0.0003 &	0.0003 &	0.0004 &	0.0003 &	0.0003 \\
 & 	1000 &	10000 &	-0.0014 &	-0.0008 &	-0.0015 &	-0.0010 &	-0.0010 &	-0.0010 &	-0.0011 &	-0.0011 &	-0.0011 &	-0.0010 \\
 & 	2000 &	5000 &	-0.0008 &	-0.0008 &	-0.0009 &	-0.0009 &	-0.0006 &	-0.0008 &	-0.0007 &	-0.0006 &	-0.0008 &	-0.0007 \\
 & 	2000 &	10000 &	0.0020 &	0.0015 &	0.0016 &	0.0013 &	0.0014 &	0.0014 &	0.0013 &	0.0015 &	0.0014 &	0.0015 \\
 & 	2000 &	20000 &	-0.0013 &	-0.0012 &	-0.0012 &	-0.0013 &	-0.0011 &	-0.0012 &	-0.0011 &	-0.0011 &	-0.0011 &	-0.0011 \\ \bottomrule
\end{tabular}
}
\end{sidewaystable}

% VARIANCE, CONT, COVDEP

\clearpage

\begin{sidewaystable}[]
\caption{Variance estimates for simulation scenarios with continuous $Z_1$ and covariate dependent sampling.} \label{table:VARIANCE_continuous_results}
\resizebox{\textwidth}{!}{
\begin{tabular}{ccccccccccccc}
\toprule
\multirow{2}{*}{\begin{tabular}[c]{@{}l@{}}Target\\ parameter\end{tabular}} & \multirow{2}{*}{\begin{tabular}[c]{@{}l@{}}Average\\ trial size\end{tabular}} & \multirow{2}{*}{$n$} & \multicolumn{10}{c}{Marginal sampling probability among non-randomized individuals}                                                                                                                   \\ \cline{4-13}
                                                                              &               &       & 0.1            & 0.2            & 0.3            & 0.4            & 0.5            & 0.6            & 0.7            & 0.7            & 0.9            & 1                                            \\ \midrule
\multirow{6}{*}{$\psi(1)$} &	1000 &	2000 &	0.0115 &	0.0082 &	0.0072 &	0.0066 &	0.0063 &	0.0062 &	0.0060 &	0.0060 &	0.0059 &	0.0059 \\
 & 	1000 &	5000 &	0.0187 &	0.0161 &	0.0153 &	0.0148 &	0.0145 &	0.0145 &	0.0143 &	0.0143 &	0.0144 &	0.0143 \\
 & 	1000 &	10000 &	0.0241 &	0.0245 &	0.0219 &	0.0220 &	0.0220 &	0.0217 &	0.0218 &	0.0217 &	0.0217 &	0.0217 \\
 & 	2000 &	5000 &	0.0064 &	0.0048 &	0.0044 &	0.0042 &	0.0040 &	0.0039 &	0.0039 &	0.0039 &	0.0038 &	0.0038 \\
 & 	2000 &	10000 &	0.0095 &	0.0080 &	0.0078 &	0.0076 &	0.0076 &	0.0074 &	0.0074 &	0.0074 &	0.0074 &	0.0074 \\
 & 	2000 &	20000 &	0.0123 &	0.0117 &	0.0114 &	0.0112 &	0.0112 &	0.0111 &	0.0111 &	0.0111 &	0.0111 &	0.0111 \\ \midrule
\multirow{6}{*}{$\psi(0)$} &	1000 &	2000 &	0.0175 &	0.0108 &	0.0090 &	0.0081 &	0.0078 &	0.0075 &	0.0073 &	0.0072 &	0.0071 &	0.0071 \\
 & 	1000 &	5000 &	0.0215 &	0.0178 &	0.0166 &	0.0155 &	0.0152 &	0.0147 &	0.0148 &	0.0146 &	0.0146 &	0.0146 \\
 & 	1000 &	10000 &	0.0259 &	0.0230 &	0.0215 &	0.0212 &	0.0210 &	0.0209 &	0.0210 &	0.0208 &	0.0208 &	0.0208 \\
 & 	2000 &	5000 &	0.0087 &	0.0061 &	0.0053 &	0.0049 &	0.0048 &	0.0046 &	0.0045 &	0.0045 &	0.0044 &	0.0044 \\
 & 	2000 &	10000 &	0.0105 &	0.0085 &	0.0078 &	0.0077 &	0.0075 &	0.0074 &	0.0074 &	0.0073 &	0.0073 &	0.0072 \\
 & 	2000 &	20000 &	0.0122 &	0.0110 &	0.0105 &	0.0103 &	0.0103 &	0.0102 &	0.0102 &	0.0102 &	0.0102 &	0.0101 \\  \midrule
\multirow{6}{*}{$\psi(1) - \psi(0)$} &	1000 &	2000 &	0.0201 &	0.0148 &	0.0133 &	0.0128 &	0.0127 &	0.0123 &	0.0121 &	0.0121 &	0.0121 &	0.0121 \\
 & 	1000 &	5000 &	0.0336 &	0.0312 &	0.0300 &	0.0289 &	0.0285 &	0.0284 &	0.0284 &	0.0282 &	0.0283 &	0.0282 \\
 & 	1000 &	10000 &	0.0469 &	0.0462 &	0.0428 &	0.0429 &	0.0429 &	0.0425 &	0.0428 &	0.0425 &	0.0426 &	0.0426 \\
 & 	2000 &	5000 &	0.0106 &	0.0088 &	0.0084 &	0.0081 &	0.0080 &	0.0079 &	0.0078 &	0.0078 &	0.0078 &	0.0078 \\
 & 	2000 &	10000 &	0.0168 &	0.0151 &	0.0148 &	0.0147 &	0.0147 &	0.0145 &	0.0145 &	0.0145 &	0.0144 &	0.0144 \\
 & 	2000 &	20000 &	0.0233 &	0.0224 &	0.0220 &	0.0218 &	0.0219 &	0.0218 &	0.0218 &	0.0217 &	0.0217 &	0.0217 \\ \bottomrule
\end{tabular}
}
\end{sidewaystable}

%VARIANCE, CONT, SRS

\clearpage

\begin{sidewaystable}[]
\caption{Variance estimates for simulation scenarios with continuous $Z_1$ and simple random sampling.} \label{table:VARIANCE_continuousSRS_results}
\resizebox{\textwidth}{!}{
\begin{tabular}{ccccccccccccc}
\toprule
\multirow{2}{*}{\begin{tabular}[c]{@{}l@{}}Target\\ parameter\end{tabular}} & \multirow{2}{*}{\begin{tabular}[c]{@{}l@{}}Average\\ trial size\end{tabular}} & \multirow{2}{*}{$n$} & \multicolumn{10}{c}{Marginal sampling probability among non-randomized individuals}                                                                                                                   \\ \cline{4-13}
                                                                              &               &       & 0.1            & 0.2            & 0.3            & 0.4            & 0.5            & 0.6            & 0.7            & 0.7            & 0.9            & 1                                            \\ \midrule
\multirow{6}{*}{$\psi(1)$} &	1000 &	2000 &	0.0089 &	0.0074 &	0.0067 &	0.0066 &	0.0062 &	0.0061 &	0.0060 &	0.0060 &	0.0059 &	0.0059 \\
 & 	1000 &	5000 &	0.0178 &	0.0163 &	0.0153 &	0.0146 &	0.0145 &	0.0148 &	0.0144 &	0.0143 &	0.0144 &	0.0143 \\
 & 	1000 &	10000 &	0.0247 &	0.0241 &	0.0223 &	0.0222 &	0.0224 &	0.0218 &	0.0217 &	0.0217 &	0.0218 &	0.0217 \\
 & 	2000 &	5000 &	0.0124 &	0.0110 &	0.0105 &	0.0106 &	0.0101 &	0.0100 &	0.0101 &	0.0099 &	0.0099 &	0.0099 \\
 & 	2000 &	10000 &	0.0255 &	0.0217 &	0.0218 &	0.0208 &	0.0209 &	0.0211 &	0.0212 &	0.0209 &	0.0207 &	0.0206 \\
 & 	2000 &	20000 &	0.0319 &	0.0296 &	0.0298 &	0.0292 &	0.0292 &	0.0290 &	0.0288 &	0.0291 &	0.0291 &	0.0288 \\ \midrule
\multirow{6}{*}{$\psi(0)$} &	1000 &	2000 &	0.0129 &	0.0094 &	0.0082 &	0.0078 &	0.0076 &	0.0074 &	0.0072 &	0.0071 &	0.0071 &	0.0071 \\
 & 	1000 &	5000 &	0.0192 &	0.0170 &	0.0156 &	0.0157 &	0.0151 &	0.0147 &	0.0149 &	0.0146 &	0.0147 &	0.0146 \\
 & 	1000 &	10000 &	0.0249 &	0.0224 &	0.0216 &	0.0211 &	0.0213 &	0.0210 &	0.0209 &	0.0208 &	0.0209 &	0.0208 \\
 & 	2000 &	5000 &	0.0142 &	0.0117 &	0.0112 &	0.0105 &	0.0107 &	0.0104 &	0.0103 &	0.0103 &	0.0103 &	0.0102 \\
 & 	2000 &	10000 &	0.0222 &	0.0203 &	0.0195 &	0.0192 &	0.0189 &	0.0189 &	0.0188 &	0.0188 &	0.0186 &	0.0186 \\
 & 	2000 &	20000 &	0.0315 &	0.0285 &	0.0275 &	0.0272 &	0.0273 &	0.0275 &	0.0273 &	0.0274 &	0.0271 &	0.0271 \\ \midrule
\multirow{6}{*}{$\psi(1) - \psi(0)$} &	1000 &	2000 &	0.0178 &	0.0142 &	0.0130 &	0.0128 &	0.0126 &	0.0123 &	0.0121 &	0.0122 &	0.0120 &	0.0121 \\
 & 	1000 &	5000 &	0.0338 &	0.0319 &	0.0297 &	0.0291 &	0.0285 &	0.0287 &	0.0285 &	0.0282 &	0.0284 &	0.0282 \\
 & 	1000 &	10000 &	0.0482 &	0.0459 &	0.0438 &	0.0431 &	0.0437 &	0.0429 &	0.0426 &	0.0425 &	0.0428 &	0.0426 \\
 & 	2000 &	5000 &	0.0244 &	0.0215 &	0.0208 &	0.0203 &	0.0202 &	0.0199 &	0.0200 &	0.0199 &	0.0198 &	0.0197 \\
 & 	2000 &	10000 &	0.0467 &	0.0417 &	0.0416 &	0.0403 &	0.0403 &	0.0404 &	0.0404 &	0.0400 &	0.0399 &	0.0398 \\
 & 	2000 &	20000 &	0.0629 &	0.0581 &	0.0575 &	0.0565 &	0.0567 &	0.0567 &	0.0564 &	0.0567 &	0.0565 &	0.0562 \\ \bottomrule
\end{tabular}
}
\end{sidewaystable}

% VARIANCE, BINARY, COVDEP

\clearpage

\begin{sidewaystable}[]
\caption{Variance estimates for simulation scenarios with binary $Z_1$ and covariate-dependent sampling.} \label{table:VARIANCE_binary_results}
\resizebox{\textwidth}{!}{
\begin{tabular}{ccccccccccccc}
\toprule
\multirow{2}{*}{\begin{tabular}[c]{@{}l@{}}Target\\ parameter\end{tabular}} & \multirow{2}{*}{\begin{tabular}[c]{@{}l@{}}Average\\ trial size\end{tabular}} & \multirow{2}{*}{$n$} & \multicolumn{10}{c}{Marginal sampling probability among non-randomized individuals}                                                                                                                   \\ \cline{4-13}
                                                                              &               &       & 0.1            & 0.2            & 0.3            & 0.4            & 0.5            & 0.6            & 0.7            & 0.7            & 0.9            & 1                                            \\ \midrule
\multirow{6}{*}{$\psi(1)$} &	1000 &	2000 &	0.0079 &	0.0056 &	0.0052 &	0.0049 &	0.0048 &	0.0047 &	0.0046 &	0.0046 &	0.0045 &	0.0045 \\
 & 	1000 &	5000 &	0.0113 &	0.0103 &	0.0097 &	0.0094 &	0.0093 &	0.0092 &	0.0091 &	0.0091 &	0.0091 &	0.0091 \\
 & 	1000 &	10000 &	0.0132 &	0.0121 &	0.0119 &	0.0118 &	0.0117 &	0.0117 &	0.0116 &	0.0116 &	0.0116 &	0.0115 \\
 & 	2000 &	5000 &	0.0041 &	0.0033 &	0.0031 &	0.0030 &	0.0029 &	0.0029 &	0.0028 &	0.0028 &	0.0028 &	0.0028 \\
 & 	2000 &	10000 &	0.0056 &	0.0050 &	0.0048 &	0.0047 &	0.0046 &	0.0046 &	0.0046 &	0.0046 &	0.0046 &	0.0045 \\
 & 	2000 &	20000 &	0.0068 &	0.0063 &	0.0062 &	0.0061 &	0.0061 &	0.0061 &	0.0060 &	0.0060 &	0.0060 &	0.0060 \\ \midrule
\multirow{6}{*}{$\psi(0)$} &	1000 &	2000 &	0.0099 &	0.0072 &	0.0065 &	0.0061 &	0.0058 &	0.0057 &	0.0056 &	0.0055 &	0.0054 &	0.0054 \\
 & 	1000 &	5000 &	0.0127 &	0.0106 &	0.0100 &	0.0095 &	0.0095 &	0.0093 &	0.0092 &	0.0091 &	0.0091 &	0.0090 \\
 & 	1000 &	10000 &	0.0153 &	0.0136 &	0.0130 &	0.0127 &	0.0126 &	0.0125 &	0.0125 &	0.0125 &	0.0124 &	0.0124 \\
 & 	2000 &	5000 &	0.0053 &	0.0039 &	0.0035 &	0.0034 &	0.0032 &	0.0031 &	0.0031 &	0.0030 &	0.0030 &	0.0030 \\
 & 	2000 &	10000 &	0.0064 &	0.0056 &	0.0052 &	0.0050 &	0.0049 &	0.0049 &	0.0049 &	0.0048 &	0.0048 &	0.0048 \\
 & 	2000 &	20000 &	0.0074 &	0.0067 &	0.0064 &	0.0063 &	0.0063 &	0.0063 &	0.0062 &	0.0062 &	0.0062 &	0.0062 \\ \midrule
\multirow{6}{*}{$\psi(1) - \psi(0)$} &	1000 &	2000 &	0.0128 &	0.0106 &	0.0099 &	0.0097 &	0.0094 &	0.0093 &	0.0093 &	0.0092 &	0.0091 &	0.0091 \\
 & 	1000 &	5000 &	0.0205 &	0.0188 &	0.0184 &	0.0180 &	0.0179 &	0.0178 &	0.0176 &	0.0176 &	0.0176 &	0.0175 \\
 & 	1000 &	10000 &	0.0263 &	0.0247 &	0.0243 &	0.0239 &	0.0239 &	0.0239 &	0.0239 &	0.0238 &	0.0238 &	0.0237 \\
 & 	2000 &	5000 &	0.0068 &	0.0060 &	0.0057 &	0.0056 &	0.0055 &	0.0054 &	0.0054 &	0.0054 &	0.0054 &	0.0053 \\
 & 	2000 &	10000 &	0.0104 &	0.0097 &	0.0095 &	0.0094 &	0.0093 &	0.0093 &	0.0093 &	0.0092 &	0.0092 &	0.0092 \\
 & 	2000 &	20000 &	0.0132 &	0.0126 &	0.0124 &	0.0123 &	0.0123 &	0.0123 &	0.0122 &	0.0122 &	0.0122 &	0.0122 \\ \bottomrule
\end{tabular}
}
\end{sidewaystable}

% VARIANCE, BINARY, SRS

\clearpage

\begin{sidewaystable}[]
\caption{Variance estimates for simulation scenarios with binary $Z_1$ and simple random sampling.} \label{table:VARIANCE_binarySRS_results}
\resizebox{\textwidth}{!}{
\begin{tabular}{ccccccccccccc}
\toprule
\multirow{2}{*}{\begin{tabular}[c]{@{}l@{}}Target\\ parameter\end{tabular}} & \multirow{2}{*}{\begin{tabular}[c]{@{}l@{}}Average\\ trial size\end{tabular}} & \multirow{2}{*}{$n$} & \multicolumn{10}{c}{Marginal sampling probability among non-randomized individuals}                                                                                                                   \\ \cline{4-13}
                                                                              &               &       & 0.1            & 0.2            & 0.3            & 0.4            & 0.5            & 0.6            & 0.7            & 0.7            & 0.9            & 1                                            \\ \midrule
\multirow{6}{*}{$\psi(1)$} &	1000 &	2000 &	0.0077 &	0.0055 &	0.0051 &	0.0049 &	0.0047 &	0.0047 &	0.0046 &	0.0046 &	0.0045 &	0.0045 \\
 & 	1000 &	5000 &	0.0111 &	0.0102 &	0.0097 &	0.0093 &	0.0092 &	0.0092 &	0.0091 &	0.0091 &	0.0091 &	0.0091 \\
 & 	1000 &	10000 &	0.0131 &	0.0121 &	0.0119 &	0.0118 &	0.0117 &	0.0117 &	0.0117 &	0.0116 &	0.0116 &	0.0115 \\
 & 	2000 &	5000 &	0.0079 &	0.0070 &	0.0067 &	0.0065 &	0.0064 &	0.0063 &	0.0063 &	0.0063 &	0.0063 &	0.0062 \\
 & 	2000 &	10000 &	0.0147 &	0.0137 &	0.0133 &	0.0130 &	0.0129 &	0.0130 &	0.0128 &	0.0129 &	0.0128 &	0.0128 \\
 & 	2000 &	20000 &	0.0180 &	0.0160 &	0.0158 &	0.0157 &	0.0157 &	0.0157 &	0.0156 &	0.0156 &	0.0155 &	0.0155 \\ \midrule
\multirow{6}{*}{$\psi(0)$} &	1000 &	2000 &	0.0092 &	0.0069 &	0.0064 &	0.0060 &	0.0058 &	0.0056 &	0.0056 &	0.0055 &	0.0054 &	0.0054 \\
 & 	1000 &	5000 &	0.0122 &	0.0102 &	0.0098 &	0.0095 &	0.0095 &	0.0093 &	0.0091 &	0.0091 &	0.0091 &	0.0090 \\
 & 	1000 &	10000 &	0.0148 &	0.0134 &	0.0130 &	0.0127 &	0.0126 &	0.0125 &	0.0125 &	0.0124 &	0.0124 &	0.0124 \\
 & 	2000 &	5000 &	0.0090 &	0.0075 &	0.0070 &	0.0069 &	0.0067 &	0.0065 &	0.0065 &	0.0064 &	0.0064 &	0.0064 \\
 & 	2000 &	10000 &	0.0139 &	0.0125 &	0.0122 &	0.0119 &	0.0117 &	0.0117 &	0.0116 &	0.0116 &	0.0116 &	0.0115 \\
 & 	2000 &	20000 &	0.0157 &	0.0150 &	0.0146 &	0.0145 &	0.0144 &	0.0144 &	0.0144 &	0.0143 &	0.0143 &	0.0143 \\ \midrule
\multirow{6}{*}{$\psi(1) - \psi(0)$} &	1000 &	2000 &	0.0126 &	0.0104 &	0.0098 &	0.0097 &	0.0095 &	0.0093 &	0.0093 &	0.0092 &	0.0091 &	0.0091 \\
 & 	1000 &	5000 &	0.0204 &	0.0186 &	0.0183 &	0.0179 &	0.0178 &	0.0177 &	0.0176 &	0.0177 &	0.0176 &	0.0175 \\
 & 	1000 &	10000 &	0.0261 &	0.0247 &	0.0243 &	0.0240 &	0.0240 &	0.0239 &	0.0239 &	0.0238 &	0.0238 &	0.0237 \\
 & 	2000 &	5000 &	0.0145 &	0.0132 &	0.0128 &	0.0127 &	0.0125 &	0.0124 &	0.0124 &	0.0123 &	0.0123 &	0.0123 \\
 & 	2000 &	10000 &	0.0273 &	0.0254 &	0.0250 &	0.0246 &	0.0244 &	0.0245 &	0.0243 &	0.0244 &	0.0244 &	0.0243 \\
 & 	2000 &	20000 &	0.0328 &	0.0303 &	0.0300 &	0.0300 &	0.0298 &	0.0298 &	0.0298 &	0.0297 &	0.0296 &	0.0296 \\ \bottomrule
\end{tabular}
}
\end{sidewaystable}

\clearpage

\begin{figure}[!htbp]
\caption{Simulation results for the sampling variance of estimators for $\psi(a), a = 0,1$ and $\psi(1) - \psi(0)$, with average trial sample size of 2000 individuals. Results in each panel are shown for different data generating mechanisms (continuous or binary $Z_1$) and sampling mechanisms (dependent on $Z_1$ or simple random sampling, SRS). In all panels, results are shown for $\widehat \psi(a)$ under marginal sampling probabilities ranging from 0.1 to 0.9, in steps of 0.1 (black markers); and for $\widehat \psi_{\text{\tiny nosub}}(a)$ under no sub-sampling (white markers). In each panel, results are shown for cohort sample sizes of 5000 (circles), 10,000 (triangles), and 20,000 (squares) individuals.}\label{fig:simulation_variance2000}
\centering
  \includegraphics[scale = 1.8]{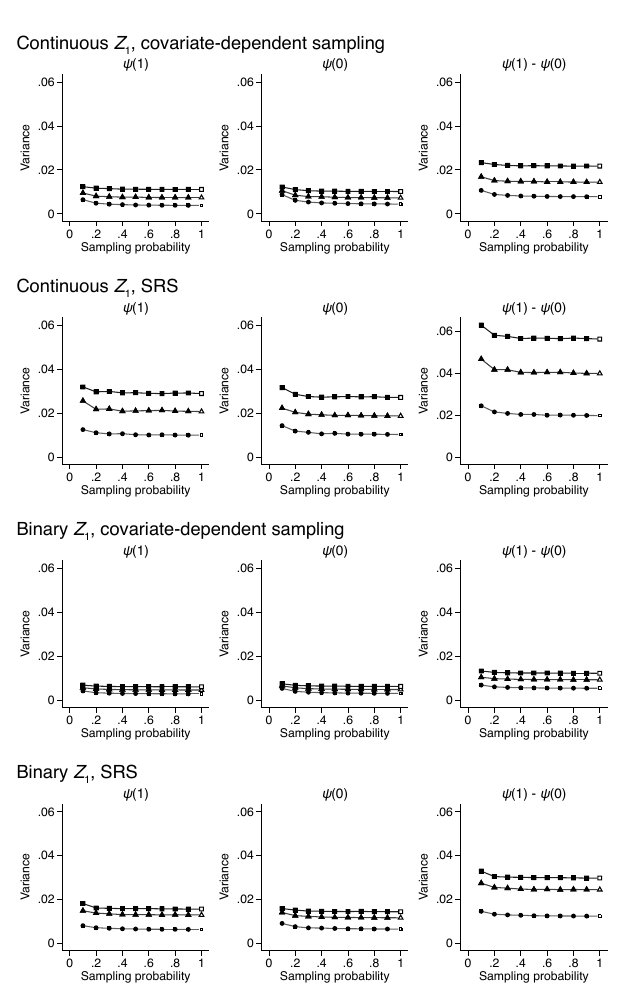}
\end{figure}

\clearpage
\begin{figure}[!htbp]
  \caption{CASS analysis results comparing the estimated standard errors of the estimator in (\ref{eq:estimator}) for $\psi(a), a = 0,1$ and $\psi(1) - \psi(0)$, under simple random sampling of non-randomized individuals, against the estimator in (\ref{eq:estimator_nosub}); see main text for details. In all panels, standard error ratios are shown for marginal sampling probabilities ranging from 0.1 to 1, in steps of 0.1.}\label{fig:cass_srs}
  \centering
  \includegraphics[scale = 1.8]{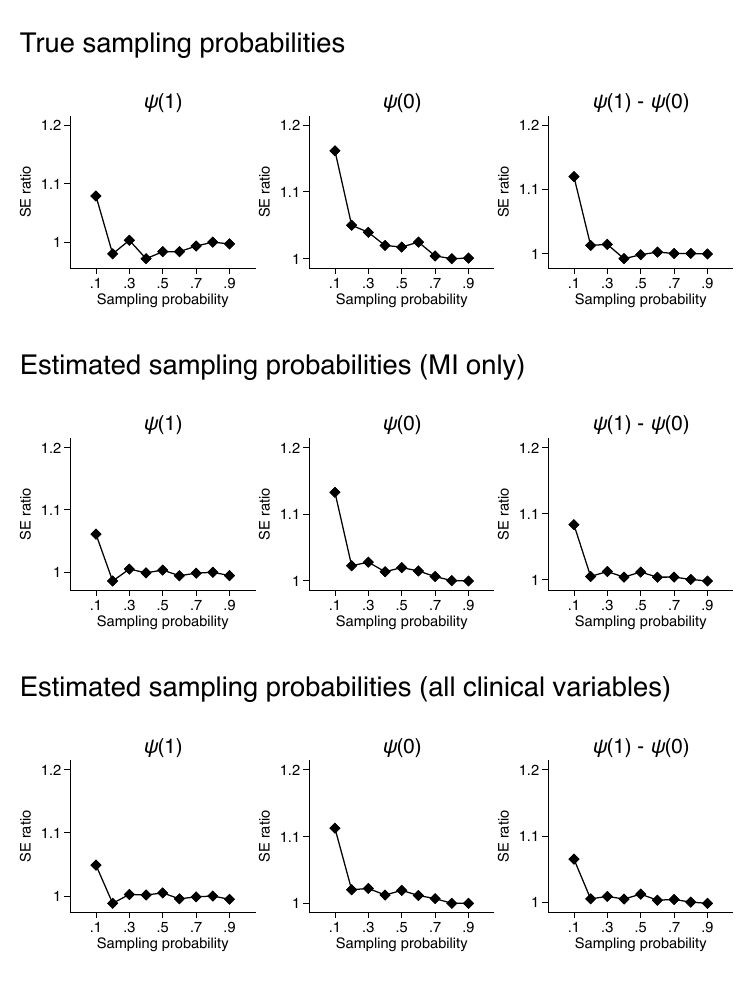}
\end{figure}

\end{document}